\documentclass[11pt,a4paper,oneside]{article}
\usepackage[top=3cm, bottom=3cm, left=2cm, right=2cm]{geometry}
\usepackage[english]{babel}
\usepackage[utf8]{inputenc}
\usepackage[T1]{fontenc}

\usepackage{graphicx} 
\usepackage{xcolor}
\usepackage{lmodern}
\usepackage{bm}
\usepackage{microtype}
\usepackage{amsthm,amsmath,amssymb,mathrsfs,dsfont,mathtools, amsfonts}
\usepackage{algorithm}
\usepackage{algpseudocode}
\usepackage{enumitem}
\usepackage[authoryear, round]{natbib}
\definecolor{dukeblue}{rgb}{0.0, 0.0, 0.61}
\usepackage[colorlinks=true, allcolors=dukeblue]{hyperref}
\usepackage{authblk}
\date{}

\title{Bayesian inference on beta diversity via feature allocation models}

\author[1]{Federica Stolf}
\author[2]{Tommaso Rigon}
\author[1]{David B. Dunson}
\affil[1]{Department of Statistical Science, Duke University, Durham, NC, USA}
\affil[2]{Department of Economics, Management, and Statistics, University of Milano–Bicocca, 20126 Milano, Italy}

\providecommand{\keywords}[1]{
  \small
  \textbf{\textit{Keywords:}} #1
  \normalsize
}

\newtheorem{theorem}{Theorem}

\newtheorem{lemma}{Lemma}
\newtheorem{proposition}{Proposition}
\theoremstyle{definition}

\newtheorem{definition}{Definition}
\newtheorem{remark}{Remark}

\newcommand{\pr}{\mathrm{pr}}

\newcommand{\V}{\mathrm{var}}

\begin{document}

\maketitle
\begin{abstract}
Beta diversity quantifies variation in species composition across ecological communities and is fundamental for understanding biodiversity patterns across space and environmental gradients. Inference is complicated by incomplete sampling, differences among communities and imperfect detection. We develop a partially exchangeable Bayesian feature allocation model for species occurrence data that accommodates unobserved species and incorporates site-specific covariates to relate variation in species composition to environmental conditions. The model yields explicit conditional distributions for species richness and sharing, together with predictions of the species newly recorded at individual sites or newly shared between sites under additional sampling. We also introduce a novel probabilistic beta diversity index with posterior uncertainty quantification. An occupancy--detection representation preserves the tractability of the occurrence model and allows sensitivity analysis under specified detectability assumptions. This tractability makes the approach scalable, enabling inference in large biodiversity studies.
Simulation studies assess predictive performance, and an application to global airborne fungal data illustrates how accounting for unobserved species and borrowing information across sites affect inference on compositional heterogeneity.
\end{abstract}


\section{Introduction}

Biodiversity studies increasingly seek to characterize how species compositions vary across ecological communities and geographic regions. Quantifying this between-community heterogeneity, commonly referred to as $\beta$-diversity, is fundamental for understanding ecological processes, monitoring environmental change, and prioritizing conservation efforts \citep{tuomisto2010diversity}. A central scientific objective is to infer the extent of species sharing across locations and to predict how many previously unobserved species would be discovered through additional sampling. Such questions arise naturally in large-scale biodiversity monitoring projects, where species occurrence data are collected across multiple habitats, ecosystems, or geographic regions.

Statistical inference on $\beta$-diversity is challenging for several reasons. Species occurrence data are typically high dimensional and sparse, with many species remaining unobserved even after extensive sampling. Ecological communities often exhibit substantial heterogeneity across locations due to differences in environmental conditions and habitat characteristics, making exchangeability assumptions unrealistic. In addition, biodiversity surveys are affected by imperfect detection: species that are present may fail to be observed because of limitations in sampling technologies and measurement protocols. Occupancy models are a standard tool in ecology for addressing imperfect detection by separating ecological and observational processes \citep{mackenzie2002estimating, Royle2005, tobler2019joint}. However, standard formulations do not jointly address species sharing across heterogeneous communities and the discovery of previously unobserved species. These challenges are particularly pronounced in DNA meta-barcoding studies.

We are motivated by the Global Spore Sampling Project (GSSP) \citep{ovaskainen2024global, abrego2024airborne}, which collected airborne fungal spore samples at locations around the world using cyclone samplers. Ecologists typically analyze these data in terms of operational taxonomic units (\texttt{OTU}s), obtained by assigning taxonomic labels to DNA sequences \citep{somervuo2016unbiased}. Because the proportion of sequences mapped to each \texttt{OTU} is not a reliable indicator of abundance, meta-barcoding studies commonly reduce the observations to binary occurrence data indicating whether a species is detected in each sample. These data are characterized by large numbers of rare species, substantial heterogeneity across locations, and non-negligible false negative rates arising from low DNA concentrations, PCR biases, and sequencing errors \citep{doi2019evaluation, schlegel2024case}.

Despite the importance of inference on $\beta$-diversity, existing approaches remain largely based on empirical dissimilarity indices with limited probabilistic interpretation and inadequate uncertainty quantification \citep{rigon2025biodiversity}. Differences in observed species overlap may reflect incomplete sampling rather than genuine variation in community composition, since species not yet recorded may nevertheless be shared across communities. The inferential goal is therefore to characterize the underlying species-sharing structure while accounting for incomplete sampling. This calls for a probabilistic framework that jointly accommodates heterogeneity across communities, incomplete sampling and imperfect detection.

Bayesian nonparametric feature allocation models provide a natural probabilistic framework for species occurrence data. Indian buffet processes \citep{Teh09, griffiths2011indian, Broderick2013two} and related models \citep{Jam17, Bat18, Mas22, Cam24, ghilFeature} characterize binary latent feature matrices in which rows correspond to samples and columns to a potentially unbounded collection of features. In biodiversity applications the features are species, so that new species emerge as sampling proceeds and the latent feature matrix encodes precisely the pattern of species sharing that $\beta$-diversity seeks to quantify. Exchangeable formulations, however, do not capture systematic variation in species composition across communities. Moreover, formulations without a detection layer attribute every zero to absence, and those without community-specific covariates cannot directly relate compositional heterogeneity to environmental conditions.
Partially exchangeable constructions have recently received growing attention for species sampling models \citep[e.g.,][]{franzolini2025multivariate, colombi2025many}, but substantially fewer results are available for feature allocations: \citet{shen2024double} proposed a bivariate beta process for genetic variants across two heterogeneous populations, though the approach is restricted to two groups and does not admit a tractable analytical characterization, while \citet{ghilotti2025bayesian} introduced a broad class of priors for partially exchangeable trait allocations, without addressing imperfect detection, including group-specific covariates or developing inferential results for biodiversity functionals such as species sharing and $\beta$-diversity.

Motivated by these challenges, we propose the class of Multi-site Occupancy-aware Species Allocation with Imperfect deteCtion (\texttt{MOSAIC}) models, which addresses all the above issues at once. Partial exchangeability allows species compositions to vary systematically across sites, while an occupancy-based detection layer distinguishes presence at the time of sampling from detection, so that a zero in the data need not imply absence in the field. A hierarchical prior allows site-specific covariates to drive species occurrence probabilities, linking compositional heterogeneity to environmental conditions while enabling borrowing of information across sites.
The number of species is finite but random, a choice that preserves substantial modeling flexibility while enabling analytical tractability and scalable posterior computation; by contrast, natural extensions of infinite-activity beta process constructions to heterogeneous populations may violate fundamental finiteness conditions or require computationally prohibitive inference procedures \citep{shen2024double}.
Our motivating application illustrates inference with more than 15,000 observed species across over 30 sites.

The resulting framework supports a coherent probabilistic treatment of diversity at every level: within-community ($\alpha$), global ($\gamma$) and between-community ($\beta$) diversity. While $\alpha$- and $\gamma$-diversity have received extensive attention in the Bayesian nonparametric literature \citep{Lijoi2007a, DeBlasi2015, ghilFeature}, species sharing in heterogeneous incidence data calls for a joint treatment across sites. We first develop the $\alpha$- and $\gamma$-diversity theory, deriving prior and posterior distributions for the site-specific and pooled species counts, together with predictions of discoveries under further sampling. These results underpin the characterization of the prior and posterior number of species shared between sites. Specifically, we derive a four-component predictive decomposition of the increase in overlap under additional sampling. We also define a $\beta$-diversity index that normalizes the expected between-site co-occurrence by the geometric mean of the within-site co-occurrences. The predictive quantities are directly relevant for biodiversity monitoring and adaptive sampling design, where decisions must be made regarding whether and where further sampling effort should be allocated.

\section{The MOSAIC model} \label{sec:pBBImodel}

\subsection{Notation and model}

We consider data sampled from $Q$ distinct sites, indexed by $q \in \{1,\dots,Q\}$; we denote by $n_q$ the number of samples from site $q$ and by $\bm{n}=(n_1,\dots,n_Q)$ the corresponding vector of sample sizes. We assume a common species reference set across sites, and we denote by $(\tilde{X}_j)_{j \ge 1}$ the sequence of all possible species labels, which are unknown a priori. Let $\mathcal{X}_q$ be the set of distinct species actually observed at site $q$ and let $K_{q,n_q}=|\mathcal{X}_q|=k_q$ be its cardinality. The total number of distinct species observed across all sites is $K_{\bm{n}}=|\bigcup_{q=1}^Q \mathcal{X}_q|=k$, and we write $X_1,\dots,X_k$ for the labels of the species actually observed, which form a subset of the complete list $(\tilde{X}_j)_{j \ge 1}$. Let $C_{n_q,n_r}=|\mathcal{X}_q \cap \mathcal{X}_r|=c_{q,r}$ denote the number of species shared between the samples from sites $q$ and $r$ with $q \neq r$. Because sites may share species, in general $K_{\bm{n}} \le \sum_{q=1}^Q K_{q,n_q}$; for instance, with $Q=2$ we have $K_{\bm{n}} = K_{1,n_1}+K_{2,n_2}-C_{n_1,n_2}$.

Throughout, we further assume that the species pool is finite, that is, the complete list reduces to $\tilde{X}_1,\dots,\tilde{X}_N$, so that the number of distinct species observed in the full sample satisfies $K_{\bm{n}} \le N$. The total number of species $N$ is itself unknown and is inferred from the data. Accordingly, we adopt the beta--Bernoulli (BB) class of feature allocation models \citep{griffiths2011indian, ghilFeature}, under which the distribution of $K_{\bm{n}}$ has bounded support given $N$. This choice is motivated by analytical tractability: known infinite-species models lead to considerably less tractable expressions \citep{shen2024double}.

The data can be encoded as binary occurrence indicators, one for each combination of sample, species, and site, and we distinguish between a version involving the full species list and its observed counterpart. Let $\tilde{\bm{Y}}=(\tilde{Y}_{ijq})$ with $i=1,\dots,n_q;\, j=1,\dots,N;\, q=1,\dots,Q$, where $\tilde{Y}_{ijq}=1$ if species $\tilde{X}_j$ is recorded in sample $i$ of site $q$, and $\tilde{Y}_{ijq}=0$ otherwise. A zero entry is ambiguous: species $\tilde{X}_j$ may be absent from site $q$ at the time of sampling, or present but not detected. We collect the \emph{observed} occurrences in $\bm{Y}=(Y_{ijq})$ with $i=1,\dots,n_q;\, j=1,\dots,k;\, q=1,\dots,Q$, obtained from $\tilde{\bm{Y}}$ by discarding the species that are never recorded, namely those $j$ for which $\tilde{Y}_{ijq}=0$ for every $i$ and $q$; equivalently, each species retained in $\bm{Y}$ has at least one non-zero entry. The number of discarded species is unknown; if all species had been discovered, $N=k$ and the two arrays would coincide. We specify a model for $\tilde{\bm{Y}}$ and derive the induced law of the observed array $\bm{Y}$.

Let $\bm{\eta}_q = (\eta_{1q},\dots,\eta_{Nq})$ and $\bm{\gamma}_q = (\gamma_{1q},\dots,\gamma_{Nq})$, for $q=1,\dots,Q$, collect independent random probabilities representing the species- and site-specific occupancy and detectability components, respectively. In our \texttt{MOSAIC} model, given these parameters, the entries of $\tilde{\bm{Y}}$ are conditionally independent Bernoulli random variables,
\begin{equation} \label{eq:bernoulli_prod}
    \tilde{Y}_{ijq} \mid \eta_{jq}, \gamma_{jq} \overset{\textup{ind}}{\sim} \mathrm{Bernoulli}(\eta_{jq} \gamma_{jq}),
\end{equation}
for $i=1,\dots,n_q$, $j=1,\dots,N$ and $q=1,\dots,Q$. The factorization in \eqref{eq:bernoulli_prod} separates the occupancy and detection components of the model: $\eta_{jq}$ is the probability that species $\tilde{X}_j$ is present at site $q$ when a sample is collected, and $\gamma_{jq}$ is the probability that it is detected, given presence. The two differ because a species present at a site may go unrecorded, so that the observed occurrences are a thinned version of the underlying presence/absence pattern. This is made precise by a data augmentation, in which each $\tilde{Y}_{ijq}$ is written as the product of two independent Bernoulli random variables representing the occupancy and detection events separately. Let $\tilde{\bm{W}}=(\tilde{W}_{ijq})$ and $\tilde{\bm{T}}=(\tilde{T}_{ijq})$, with $i=1,\dots,n_q$, $j=1,\dots,N$ and $q=1,\dots,Q$, be augmented binary arrays encoding respectively the true presence of each species and its potential detection, so that $\tilde{W}_{ijq} \mid \eta_{jq} \overset{\textup{ind}}{\sim} \mathrm{Bernoulli}(\eta_{jq})$ and $\tilde{T}_{ijq} \mid \gamma_{jq} \overset{\textup{ind}}{\sim} \mathrm{Bernoulli}(\gamma_{jq})$. Then $\tilde{Y}_{ijq} = \tilde{W}_{ijq}\tilde{T}_{ijq}$, so that a species is recorded if and only if it is both present and detected.

To accommodate unseen species, we treat the total species pool size as random, with $N \sim \mathrm{Poisson}(\lambda)$ and $\lambda>0$.
Moreover, we assume that the parameters $\gamma_{jq}$ and $\eta_{jq}$ are conditionally independent for $j=1,\dots,N$ and $q=1,\dots,Q$ given a set of parameters $(\alpha_q, \theta_q, \sigma_q)$ and follow a beta distribution. We will further denote the corresponding vectors of parameters as
$ \bm{\alpha} = (\alpha_1,\dots, \alpha_Q)$, $\bm{\theta} = (\theta_1, \dots, \theta_Q)$ and $\bm{\sigma} = (\sigma_1, \dots, \sigma_Q)$.
This implies the following specification for the \texttt{MOSAIC} model with parameters $(\lambda, \bm{\alpha}, \bm{\theta}, \bm{\sigma})$
\begin{equation} \label{eq:mu_BBimperfect}
\begin{aligned}
     N&\sim \mathrm{Poisson}(\lambda),\\
    \eta_{jq} &\overset{\textup{ind}}{\sim} \mathrm{Beta}\{-\alpha_q, \sigma_{q}(\alpha_q+\theta_q)\}, \quad (j=1,\dots,N, q=1,\dots,Q),\\
    \gamma_{jq} &\overset{\textup{ind}}{\sim} \mathrm{Beta}\{-\alpha_q + \sigma_{q}(\alpha_q+\theta_q), (1-\sigma_q)(\alpha_q+\theta_q)\}, \quad (j=1,\dots,N, q=1,\dots,Q),
\end{aligned}
\end{equation}
where $\sigma_q \in (0,1)$, $\alpha_q<0$, $\theta_q>-\alpha_q$ for $q=1,\dots,Q$.
Define the occurrence probabilities $\pi_{jq} = \eta_{jq}\gamma_{jq}$ for $j=1,\dots,N$ and $q=1,\dots,Q$, and collect these probabilities into the vectors $\bm{\pi}_q = (\pi_{1q},\dots,\pi_{Nq})$ for each site $q$.
By a result of \citet{james1972products} on products of independent beta random variables,
\begin{equation}\label{eq:marginal_beta}
\pi_{jq} \overset{\textup{ind}}{\sim} \mathrm{Beta}(-\alpha_q,\,  \alpha_q + \theta_q), \quad (j=1,\dots,N,\ q=1,\dots,Q),
\end{equation}
and in particular $E(\pi_{jq}) = -\alpha_q/\theta_q$. Thus, the occupancy and detection layers provide a data-augmentation representation of this marginal model. Crucially, \eqref{eq:marginal_beta} holds for every value of $\sigma_q\in(0,1)$: the parameter $\sigma_q$ redistributes mass between occupancy and detectability. Its role is made explicit by the expectations
\begin{equation*}
    E(\eta_{jq}) = \frac{-\alpha_q}{\sigma_q(\alpha_q+\theta_q)-\alpha_q}, \qquad E(\gamma_{jq}) = \frac{-\alpha_q+\sigma_q(\alpha_q+\theta_q)}{\theta_q},
\end{equation*}
which are respectively decreasing and increasing in $\sigma_q$. Values of $\sigma_q$ close to $1$ thus correspond to near-perfect detectability, whereas values close to $0$ describe the least favorable regime, in which every species occupies every site and the observed sparsity is ascribed entirely to detection failure. In the limit $\sigma_q\to1$ the law of $\gamma_{jq}$ degenerates at one, giving $E(\eta_{jq}) = -\alpha_q/\theta_q$, which recovers the expected occurrence probability of a partially exchangeable BB model with perfect detection.

This framework borrows information across sites in two ways. First, all sites share a common total number of species $N$. Second, it introduces dependence among species occurrence probabilities through site-specific covariates via a hierarchical formulation, detailed in Section \ref{sec:prior_specify}.

\begin{remark}\label{rem:identifiability}
The occupancy indicators $\tilde{W}_{ijq}$ are allowed to vary across samples $i$ within a site, so that occupancy refers to presence at the time sample $i$ was collected rather than to a fixed attribute of the site. Standard occupancy models \citep[e.g.,][]{mackenzie2002estimating, Wenger2008} instead assign a single indicator to all samples from the same site. We depart from this convention primarily for reasons of tractability: sample-level indicators lead to remarkably simple expressions for the quantities of interest derived in Section~\ref{sec:theory} and Section~\ref{sec:distr_theory}. Under this specification, the observed occurrences $\tilde{\bm{Y}}$ depend on $\eta_{jq}$ and $\gamma_{jq}$ only through their product $\pi_{jq} = \eta_{jq}\gamma_{jq}$. As a consequence, occupancy and detectability are not separately identifiable, and neither is $\sigma_q$, which determines how $\pi_{jq}$ is split between the two components without affecting the law of $\tilde{\bm{Y}}$. The probability $\pi_{jq}$, however, is identified, and it is the focus of our primary analyses. We therefore do not attempt to estimate $\sigma_q$ and treat it instead as a sensitivity parameter, assessing how the conclusions change across detectability scenarios ranging from perfect detection ($\sigma_q \to 1$) to substantial under-detection.
\end{remark}

\section{Biodiversity quantification and inference}  \label{sec:theory_beta}

\subsection{Species richness and shared species} \label{sec:theory}

In this section, we investigate the theoretical properties of \texttt{MOSAIC} for key quantities of interest in biodiversity studies.
Among the many diversity measures, species richness, the total number of species in a community, is arguably the simplest and most widely used. In the multi-site setting a distinction is needed: the richness of a single site, or $\alpha$-diversity, refers to one site at a time, whereas the richness across sites, or $\gamma$-diversity, refers to the pooled species list. The results below concern both, together with the number of shared species, a priori and a posteriori. The study of species richness is well established in the Bayesian nonparametric literature for exchangeable data: see \citet{Lijoi2007a, DeBlasi2015, zito2023bayesian, rigon2025biodiversity} for abundance data and \citet{ghilFeature} for incidence data. In the multi-site case, that is, under partial exchangeability, closed-form expressions are instead scarce, with the notable exception of \citet{colombi2025many} for abundance data. The findings below are the first available in the
incidence data literature.

We present results based on the observed data $\bm{Y}$; analogous results for the latent occupancies $\bm{W}$ are provided in Section~\ref{sec:supp_TheLatent} of the Supplementary Material.
These latent-occupancy results provide insights into how assumptions about detectability affect inference on species richness and shared species.
Recall that $K_{q,n_q}$ and $K_{\bm{n}}$ denote the number of distinct species observed at site $q$ and across all sites, respectively, so that $K_{\bm{n}}  \le N$. In the terminology above, $K_{q,n_q}$ measures $\alpha$-diversity at site $q$, whereas $K_{\bm{n}}$ measures $\gamma$-diversity across all sites. Both depend on sampling effort and are bounded above by the species pool size $N$. In ecological applications, the expectations $E(K_{\bm{1}}), \dots, E(K_{\bm{n}})$ provide model-based estimates of the global rarefaction curve \citep{gotelli2001biodiveristy, zito2023bayesian} in terms of the global sampling effort $\bm{n}$. Analogously, the sequence $E(K_{q,1}), \dots, E(K_{q,n_q})$ defines a model-based estimate of the rarefaction curve at site $q$ as a function of the sampling effort $n_q$.

The quantity $p_{0,q}(\alpha_q, \theta_q, n_q) = (\alpha_q + \theta_q)_{n_q}/(\theta_q)_{n_q} = E\{(1 - \pi_{jq})^{n_q}\}$ is the probability that a species goes unrecorded in the $n_q$ samples from site $q$, and $p_0(\bm{\alpha}, \bm{\theta}, \bm{n}) = \prod_{q=1}^Q p_{0,q}(\alpha_q, \theta_q, n_q)$ is the probability that it goes unrecorded at every site. Conversely, we let $p_{1,q}(\alpha_q, \theta_q, n_q) := 1 - p_{0,q}(\alpha_q, \theta_q, n_q)$ be the probability that a species is recorded at least once at site $q$, and $p_{1}(\bm{\alpha}, \bm{\theta}, \bm{n}) := 1 - p_{0}(\bm{\alpha}, \bm{\theta}, \bm{n})$ the probability that it is recorded at least once overall. The following proposition provides the marginal distributions of these quantities.

\begin{proposition}[\emph{A priori} distinct species] \label{pr:Kprior}
Under the \texttt{MOSAIC} model with parameters $(\lambda, \bm{\alpha}, \bm{\theta})$, the number of distinct species observed at site $q$ and across sites is distributed as
\begin{equation*}
K_{q,n_q} \sim \mathrm{Poisson}\{\lambda \, p_{1,q}(\alpha_q, \theta_q, n_q) \}, \quad K_{\bm{n}} \sim \mathrm{Poisson}\{\lambda \, p_1(\bm{\alpha}, \bm{\theta}, \bm{n}) \}, \quad (q=1,\dots,Q).
\end{equation*}
\end{proposition}

Proposition~\ref{pr:Kprior} gives a simple interpretation of the expected species counts. In $E(K_{\bm{n}}) = \lambda\, p_1(\bm{\alpha}, \bm{\theta}, \bm{n})$, the parameter $\lambda$ is the expected total number of species across sites and $p_1(\bm{\alpha}, \bm{\theta}, \bm{n})$ is the fraction of these that is observed under sampling effort $\bm{n}$. The site-specific quantities admit the same reading, since $E(K_{q,n_q}) = \lambda\, p_{1,q}(\alpha_q, \theta_q, n_q)$.

Of key ecological interest in the study of community similarity is the number of species shared between two sites $q$ and $r$, denoted by $C_{n_q,n_r}$. When $Q=2$, the number of shared species is linked to the site-specific and total counts through $C_{n_1,n_2} = K_{1,n_1} + K_{2,n_2} - K_{\bm{n}}$. The following proposition provides the distribution of the number of species shared between any pair of sites.

\begin{proposition}[\emph{A priori} shared species] \label{pr:shared}
Under the \texttt{MOSAIC} model with parameters $(\lambda, \bm{\alpha}, \bm{\theta})$, the number of species shared between sites $q$ and $r$, with $q \neq r$, is distributed as
\begin{equation*}
C_{n_q,n_r} \sim \mathrm{Poisson}\{\lambda \, p_{1,q}(\alpha_q, \theta_q, n_q) \, p_{1,r}(\alpha_r, \theta_r, n_r)\}.
\end{equation*}
\end{proposition}

Proposition~\ref{pr:shared} is the first result characterizing the distribution of the number of shared species for partially exchangeable feature allocation models with finitely many features. The mean of $C_{n_q,n_r}$ is the expected species richness $\lambda$ multiplied by the probability that a species is recorded at both sites, which factorizes because the occurrence probabilities are independent across sites.

The above result allows us to define model-based estimates of rarefaction for shared species between pairs of sites, providing a measure of overlap in species composition. Although the definition of a single-site rarefaction curve is straightforward, its extension to shared species is less direct. The shared rarefaction surface is the matrix of dimension $n_q \times n_r$ with entries $E(C_{i_q,i_r})$, for $i_q = 1,\dots,n_q$ and $i_r = 1,\dots,n_r$, representing the expected number of shared species as sampling effort increases at both sites. This two-dimensional surface can be summarized by the one-dimensional curve $E(C_{1,1}), E(C_{2,2}), \dots, E(C_{\bar{n}, \bar{n}})$, with $\bar{n} = \min(n_q,n_r)$, obtained by increasing sampling effort jointly and equally at the two sites. This curve summarizes how overlap accumulates under balanced sampling. When sample sizes differ, it stops at $\bar{n}$; the full surface captures the remaining range of sampling effort.

Propositions~\ref{pr:Kprior}-\ref{pr:shared} concern \emph{a priori} properties of $K_{q,n_q}$, $K_{\bm{n}}$ and $C_{n_q,n_r}$. We now consider the important problem of predicting the number of new species that would be discovered if $\bm{s} = (s_1,\dots,s_Q)$ additional samples were collected after an initial $\bm{n}$, resulting in an enlarged sample of size $\bm{n} + \bm{s}$.
We denote by $K_{q,s_q}^{(n_q)}$ and $K_{\bm{s}}^{(\bm{n})}$ the number of new species that would be discovered at site $q$ and across all sites, respectively. These quantities extrapolate the accumulation curves beyond the observed effort, since $E(K_{\bm{n}+\bm{s}} \mid \bm{Y}) = k + E(K_{\bm{s}}^{(\bm{n})} \mid \bm{Y})$. Two groups of species contribute to the predictive of $K_{q,s_q}^{(n_q)}$: those not yet recorded anywhere, whose number is random; and the $k - k_q$ species recorded at some other site but not at $q$, whose number is known given $\bm{Y}$.

\begin{theorem}[\emph{A posteriori} distinct species] \label{thm:Kpred}
Under the \texttt{MOSAIC} model with parameters $(\lambda, \bm{\alpha}, \bm{\theta})$, the number of new species discovered at site $q$ and across sites satisfies
\begin{equation*}
\begin{aligned}
K_{q,s_q}^{(n_q)} \mid \bm{Y} \overset{d}{=} \; & \mathrm{Poisson}\{\lambda \, p_{0}(\bm{\alpha}, \bm{\theta}, \bm{n}) \, p_{1,q}(\alpha_q, \theta_q + n_q, s_q)\} \\
 \; &+ \mathrm{Binomial}\{k - k_q, \, p_{1,q}(\alpha_q, \theta_q + n_q, s_q)\}, \quad (q = 1,\dots,Q),\\
K_{\bm{s}}^{(\bm{n})} \mid \bm{Y} \; \sim \; & \mathrm{Poisson}\{\lambda \, p_{0}(\bm{\alpha}, \bm{\theta}, \bm{n}) \, p_{1}(\bm{\alpha}, \bm{\theta} + \bm{n}, \bm{s})\}.
\end{aligned}
\end{equation*}
\end{theorem}

Here and in the following, the sum of two named distributions denotes the distribution of the sum of two independent random variables with those laws. Theorem~\ref{thm:Kpred} admits the same interpretation as its \emph{a priori} counterpart, with the roles of the two factors made explicit: $\lambda \, p_{0}(\bm{\alpha}, \bm{\theta}, \bm{n})$ is the expected number of species left undiscovered after the initial $\bm{n}$ samples, and $p_{1}(\bm{\alpha}, \bm{\theta} + \bm{n}, \bm{s})$ is the probability that one such species is recorded at least once in the $\bm{s}$ additional samples. The Bayesian estimator of the number of new species is therefore $E(K_{\bm{s}}^{(\bm{n})} \mid \bm{Y}) = \lambda \, p_{0}(\bm{\alpha}, \bm{\theta}, \bm{n}) \, p_{1}(\bm{\alpha}, \bm{\theta} + \bm{n}, \bm{s})$, whereas the posterior expected number of species not yet recorded at site $q$ is $\lambda\,p_0(\bm{\alpha},\bm{\theta},\bm{n})+k-k_q$. The latter includes species already recorded elsewhere, so that $K_{q,s_q}^{(n_q)}$ counts species new to site $q$ rather than new to the study.

Finally, we characterize the predictive distribution of the number of new species shared between two sites $q$ and $r$, denoted by $C_{s_q,s_r}^{(n_q,n_r)}=C_{n_q+s_q,n_r+s_r}-C_{n_q,n_r}$. Let $k_{q,r} = k_q + k_r - c_{q,r}$ be the number of distinct species recorded at site $q$ or at site $r$, then three groups of species contribute to the predictive: those not yet recorded at either site, which must appear at both; those recorded only at site $q$, which must appear at site $r$; and those recorded only at site $r$, which must appear at site $q$.

\begin{theorem}[\emph{A posteriori} shared species] \label{thm:pred_shared}
Under the \texttt{MOSAIC} model with parameters $(\lambda, \bm{\alpha}, \bm{\theta})$, the number of species shared between sites $q$ and $r$, with $q \neq r$, that would be discovered in $(s_q,s_r)$ additional samples satisfies
\begin{equation*}
\begin{aligned}
C_{s_q,s_r}^{(n_q,n_r)} \mid \bm{Y} \overset{d}{=} \; & \mathrm{Poisson}\{\lambda \, p_0(\bm{\alpha}, \bm{\theta}, \bm{n}) \, p_{1,q}(\alpha_q, \theta_q + n_q, s_q) \, p_{1,r}(\alpha_r, \theta_r + n_r, s_r)\} \\
+ \; & \mathrm{Binomial}\{k - k_{q,r}, \, p_{1,q}(\alpha_q, \theta_q + n_q, s_q) \, p_{1,r}(\alpha_r, \theta_r + n_r, s_r)\} \\
+ \; & \mathrm{Binomial}\{k_q - c_{q,r}, \, p_{1,r}(\alpha_r, \theta_r + n_r, s_r)\}  \\
+ \; & \mathrm{Binomial}\{k_r - c_{q,r}, \, p_{1,q}(\alpha_q, \theta_q + n_q, s_q)\},
\end{aligned}
\end{equation*}
where the four terms are independent.
\end{theorem}

Theorem~\ref{thm:pred_shared} provides a simple closed-form expression for extrapolating the shared rarefaction surface, and is the first result analytically characterizing the posterior
distribution of the number of shared species for partially exchangeable feature allocation models. Each term corresponds to a distinct group of species. The Poisson and the first binomial term account for species not yet recorded at either site, which must appear at both: the former counts those never recorded anywhere, whose number is itself random and Poisson distributed, whereas the latter counts the $k - k_{q,r}$ species already recorded at some other site, whose number is known given $\bm{Y}$. The last two binomial terms count species recorded at only one of the two sites, which need appear at the other site alone.

All the above results extend to the case in which the total number of species in \eqref{eq:mu_BBimperfect} follows a negative binomial distribution, which accommodates overdispersion while retaining the same analytical tractability and interpretability. Section \ref{sec:supp_nb} of the Supplementary Material reports the details.

\subsection{A model-based index of $\beta$-diversity}\label{sec:betadiv}

To summarize heterogeneity in species composition, we now turn to $\beta$-diversity. While the pairwise rarefaction and extrapolation curves of Section~\ref{sec:theory} based on shared species are valuable, it is often convenient to summarize heterogeneity across regions in a single number. Many measures of $\beta$-diversity have been introduced in ecology, but there is no overall consensus on which one is the most appropriate \citep{anderson2011navigating}. The oldest is the decomposition of \citet{Whittaker1960}, in which the $\gamma$-diversity of a set of sites is expressed as the product of the average $\alpha$-diversity of the individual sites and a multiplicative factor, the $\beta$-diversity, measuring how much the pooled species list exceeds what a typical site contains. A second family of measures is instead based on dissimilarities; for incidence data the most popular are the Jaccard and Sørensen indices \citep{tuomisto2010diversity}. Jaccard similarity is the number of shared species divided by the total number of species, namely $|\mathcal{X}_q \cap \mathcal{X}_r|/|\mathcal{X}_q \cup \mathcal{X}_r|$ in our notation, for $q \neq r$; Sørensen is similar but places more weight on the overlap, so that for the same two sets it never returns a lower value than Jaccard. \citet{barwell2015measuring} point out that many commonly used $\beta$-diversity indices are highly correlated and redundant, yet do not properly account for undetected shared and unshared species.
Our proposed \texttt{MOSAIC} model offers a principled way to solve this problem,
based on a ratio of expected shared-species counts between and within sites.

\begin{definition}
    The $\beta$-diversity between sites $q$ and $r$, with $q \neq r$, is
\begin{equation*}
    \beta_{qr} = 1- \frac{\sum_{j=1}^N \eta_{jq} \eta_{jr} }{\big(\sum_{j=1}^N \eta_{jq}^2 \sum_{j=1}^N \eta_{jr}^2\big)^{1/2}}.
\end{equation*}
\end{definition}

\noindent The proposed index is the cosine dissimilarity between the occupancy profiles of the two sites; it is therefore symmetric, $\beta_{qr} = \beta_{rq}$, and, since the occupancy probabilities are nonnegative, conveniently bounded between $0$ and $1$. It equals $0$ when the two profiles are proportional, and $1$ when no species has positive occupancy probability at both sites; in general, increasing values of $\beta_{qr}$ indicate greater heterogeneity between sites. The numerator of $1-\beta_{qr}$ is the expected number of species shared in two samples drawn from \emph{different} sites, say the $i$th sample from site $q$ and the $i'$th from site~$r$,
\begin{equation*}
    \sum_{j=1}^N \pr(\tilde{W}_{ijq}=\tilde{W}_{i'jr} = 1\mid \bm{\eta}_q, \bm{\eta}_r) = \sum_{j=1}^N \eta_{jq}\eta_{jr},
\end{equation*}
whereas each term in the denominator is the expected number of species shared in two samples $i \neq i'$ drawn from the \emph{same} site,
\begin{equation*}
    \sum_{j=1}^N \pr(\tilde{W}_{ijq}=\tilde{W}_{i'jq} = 1\mid \bm{\eta}_q) = \sum_{j=1}^N \eta_{jq}^2 .
\end{equation*}
The denominator thus rescales the between-site co-occurrence by the geometric mean of the two within-site co-occurrences, so that the index measures similarity of occupancy patterns rather than their overall magnitude: two sites whose profiles differ by a common multiplicative factor are regarded as maximally similar.

The index $\beta_{qr}$ provides an inferential counterpart to the empirical Ochiai dissimilarity \citep{Ochiai1957}. Indeed, replacing each occupancy probability $\eta_{jq}$ with the observed presence indicator $\mathds{1}(\tilde{X}_j \in \mathcal{X}_q)$ yields one minus the Ochiai similarity coefficient, $1 - |\mathcal{X}_q \cap \mathcal{X}_r|/\sqrt{|\mathcal{X}_q|\,|\mathcal{X}_r|}$, since the numerator counts shared species and the denominator is the geometric mean of the two observed species counts. The Jaccard similarity coefficient differs in that it normalizes the shared species count by $|\mathcal{X}_q \cup \mathcal{X}_r|$. Using latent occupancy probabilities extends this empirical comparison to account for unobserved species and enables posterior uncertainty quantification.
A further connection is with the correlation between random probability measures in Proposition~4 of \citet{franzolini2025multivariate}, which equals the probability of a tie across two groups, normalized by the geometric mean of the two within-group tie probabilities: our index has the same structure, with co-occurrences of binary features replacing ties among species labels, and thus provides its feature allocation counterpart. Their quantity would in turn provide a natural definition of $\beta$-diversity for multivariate species sampling models, based on relative abundances rather than on presence--absence patterns.

Our definition is based on the occupancy probabilities $\bm{\eta}_q$ rather than on the occurrence probabilities $\bm{\pi}_q$, since under-detection thins the recorded occurrences and would otherwise bias the resulting measure of biodiversity; it is moreover defined over the complete list of $N$ species, and hence accounts for the species left unobserved. Since $\beta_{qr}$ is a function of occupancy probabilities whose posterior is available, it admits full uncertainty quantification in addition to point estimation. However, as noted in Remark~\ref{rem:identifiability}, the data inform the products $\pi_{jq}=\eta_{jq}\gamma_{jq}$ but do not separately identify occupancy and detectability; posterior inference for $\bm{\eta}_q$ therefore depends on the choice of $\sigma_q$. The cosine similarity is invariant to a common rescaling of each occupancy profile, but varying $\sigma_q$ need not produce such a rescaling. Sensitivity of $\beta_{qr}$ to detectability assumptions must therefore be assessed empirically. Finally, the same construction applies more broadly to any probabilistic framework that yields species-occurrence probabilities, such as joint species distribution models \citep{OvaskainenEtal2016}, on replacing $\bm{\eta}_q$ by $\bm{\pi}_q$ and $N$ by $k$; the resulting index no longer corrects for under-detection or for unseen species, but remains well defined.

\subsection{Hierarchical prior specification with site-specific covariates} \label{sec:prior_specify}

We incorporate site-specific covariates through hierarchical priors on $\bm{\alpha}$ and $\bm{\theta}$. Using the marginal beta representation in \eqref{eq:marginal_beta}, we reparametrize in terms of the mean and precision,
\begin{equation} \label{eq:rip_beta}
    \mu_q = \frac{-\alpha_q}{\theta_q}, \quad \phi_q = \theta_q, \quad \text{with } \mu_q \in (0,1), \ \phi_q>0,
\end{equation}
for $q=1,\dots,Q$. Therefore, $\pi_{jq} \mid \mu_q, \phi_q \sim \mathrm{Beta}\{\mu_q\phi_q, (1-\mu_q)\phi_q\}$, with
\begin{equation*}
    E(\pi_{jq}\mid\mu_q, \phi_q) = \mu_q, \quad \V(\pi_{jq}\mid\mu_q, \phi_q) = \frac{\mu_q(1-\mu_q)}{\phi_q+1}.
\end{equation*}
In biodiversity studies, it is common to have site-specific covariates, such as mean wind speed, mean temperature, and mean precipitation. We incorporate this information into a regression model for $\mu_q$. Letting $\bm{\zeta}=(\zeta_1, \dots, \zeta_H)^{\top}$ denote regression coefficients and $\bm{z}_q=(z_{q1}, \dots, z_{qH})^{\top}$ site-specific covariates, we assume
\begin{equation} \label{eq:linpred_h}
    \log\bigg(\frac{\mu_q}{1-\mu_q}\bigg) = \bm{z}_q^{\top} \bm{\zeta}, \quad \bm{\zeta} \sim \mathcal{N}(\bm{b}_0, \bm{B}_0),
\end{equation}
for $q=1,\dots,Q$, and we let $\bm{Z}$ denote the $Q\times H$ covariate matrix. This specification admits a direct interpretation of the regression coefficients in terms of odds ratios, analogous to standard logistic regression. For the precision parameters $\bm{\phi}=(\phi_1, \dots, \phi_Q)^{\top}$, we assume a common distribution across sites, $\phi_q \sim \mathrm{Gamma}(a_{\phi}, b_{\phi})$ for $q=1,\dots,Q$.
For the parameter of the Poisson distribution for $N$, we assume $\lambda \sim \mathrm{Gamma}(\rho_0, \rho_0/ \nu_0)$, which implies $N \sim \mathrm{NegBinomial}(\rho_0,\nu_0)$.
We recall that all the theoretical results in Section~\ref{sec:theory} extend to the negative binomial case, as detailed in the Supplementary Material.

\subsection{Posterior computation}\label{sec:distr_theory}

For posterior computation, the representation in \eqref{eq:marginal_beta} allows us to analytically marginalize over the species-specific probabilities and the unobserved species, yielding a closed-form likelihood for $(\lambda,\bm{\alpha},\bm{\theta})$.
Specifically, the marginal law of $\bm{Y}$ under \texttt{MOSAIC}, or more precisely the joint law of $(\bm{Y} = \bm{y}, K_{\bm{n}} = k)$, depends on $\bm{y}$ only through the sufficient statistics given by the marginal occurrence frequencies of each species. Let $M_{jq} = \sum_{i=1}^{n_q} Y_{ijq}$ denote the frequency of species $X_j$ at site $q$, let $\bm{M} = (M_{jq})$ be the
corresponding matrix, and let $\bm{m} = (m_{jq})$ be its observed counterpart, for $j=1,\dots,k$ and $q=1,\dots,Q$. Then
\begin{equation}
\label{eq:pEFPF}
\begin{split}
\pr(\bm{Y} = \bm{y}, K_{\bm{n}} = k \mid \lambda, \bm{\alpha}, \bm{\theta}) = \; & \frac{\lambda^k}{k!}
\exp\bigg(-\lambda \bigg\{ 1 - \prod_{q=1}^Q
\frac{(\alpha_q+\theta_q)_{n_q}}{(\theta_q)_{n_q}} \bigg\}\bigg) \\
& \times \prod_{q=1}^Q \bigg[ \bigg\{\frac{-\alpha_q}{(\theta_q)_{n_q}}\bigg\}^k
\prod_{j=1}^k (1-\alpha_q)_{m_{jq}-1}  (\alpha_q + \theta_q)_{n_q - m_{jq}}
\bigg],
\end{split}
\end{equation}
where $(x)_m = \Gamma(x+m)/\Gamma(x)$ is the Pochhammer symbol and $\Gamma(x)$ is the gamma function. The derivation is provided in Section \ref{sec:supp_marginal} of the Supplementary Material. We refer to \eqref{eq:pEFPF} as a \emph{partially exchangeable feature probability function} (pEFPF). This quantity is related to the partially exchangeable partition probability function of \citet{franzolini2025multivariate} and, for fixed hyperparameters, specializes the partially exchangeable trait probability function of \citet{ghilotti2025bayesian} to the case of binary traits. The posterior results established in the latter work also yield closed-form distributions for the occurrence probabilities $\pi_{jq}$ and the total number of species $N$, conditional on the hyperparameters.
These conditional results remain available under the hierarchical regression specification introduced in the previous section, where site-specific covariates drive the occurrence probabilities. Posterior inference for the full hierarchical model proceeds through an efficient Markov chain Monte Carlo (MCMC) scheme, where the availability of conditional distributions in closed form yields substantial computational savings. Details are provided in Section~\ref{sec:computation} of the Supplementary Material.

\section{Simulation studies} \label{sec:sim_study}

We use synthetic data to assess estimation of the regression coefficients $\bm{\zeta}$ and precision parameters $\bm{\phi}$, and prediction of shared and site-specific species counts. For site-specific discoveries, $K_{q,s_q}^{(n_q)}$, we use the exchangeable negative binomial mixture of beta--Bernoulli models of \citet{ghilFeature} as a benchmark.

We begin by evaluating the performance of the proposed approach in estimating the number of shared species in a simple setting with $Q=3$ sites.
We generate occurrence data under the hierarchical formulation of the \texttt{MOSAIC} model with $N=3{,}000$, regression coefficients $\bm{\zeta}=(-3, 0.5, -0.3)^{\top}$, where the negative intercept reflects the presence of many rare species, and draw $\bm{\phi}$ from a Gamma distribution with shape 4 and rate 2. The elements of the design matrix $\bm{Z}$ are simulated from standard normal variables.
We consider sample sizes $\bm{n}=(100,80,80)$ and split the data into training and test sets, using $60\%$ of the observations in each site for training.
The \texttt{MOSAIC} model is fitted to the training data using the following hyperparameter specification: $\bm{b}_0 = (-2,0,0)^{\top}$, $\bm{B}_0 = \bm{I}$, $\rho_0 = 10$, $\nu_0 = 1.5k$, $a_{\phi}=2$ and $b_{\phi}=2$. The MCMC algorithm is run for $8{,}000$ iterations, of which the first half is discarded as burn-in. We predict the number of shared species $C_{s_q,s_r}^{(n_q,n_r)}$ for each pair of sites in the test set, with $\bm{s} = (40,32,32)$.

Figure \ref{fig:sim_shared} displays the rarefaction surfaces for shared species on the test set, together with the corresponding model-based predictions.
The posterior mean closely follows the empirical surfaces for all pairs, indicating a strong agreement between observed and predicted values. For the pair with the same number of samples, $(2,3)$, we additionally report the full one-dimensional rarefaction curve, corresponding to the diagonal of the surface, together with its model-based prediction on the test set. This curve shows how shared-species counts increase under balanced sampling effort. Figure \ref{fig:supp_surfUQ} in the Supplementary Material shows the $95\%$ posterior credible intervals for the predicted rarefaction surfaces for all site pairs.

\begin{figure}[t!]
    \centering
    \includegraphics[width=0.8\textwidth]{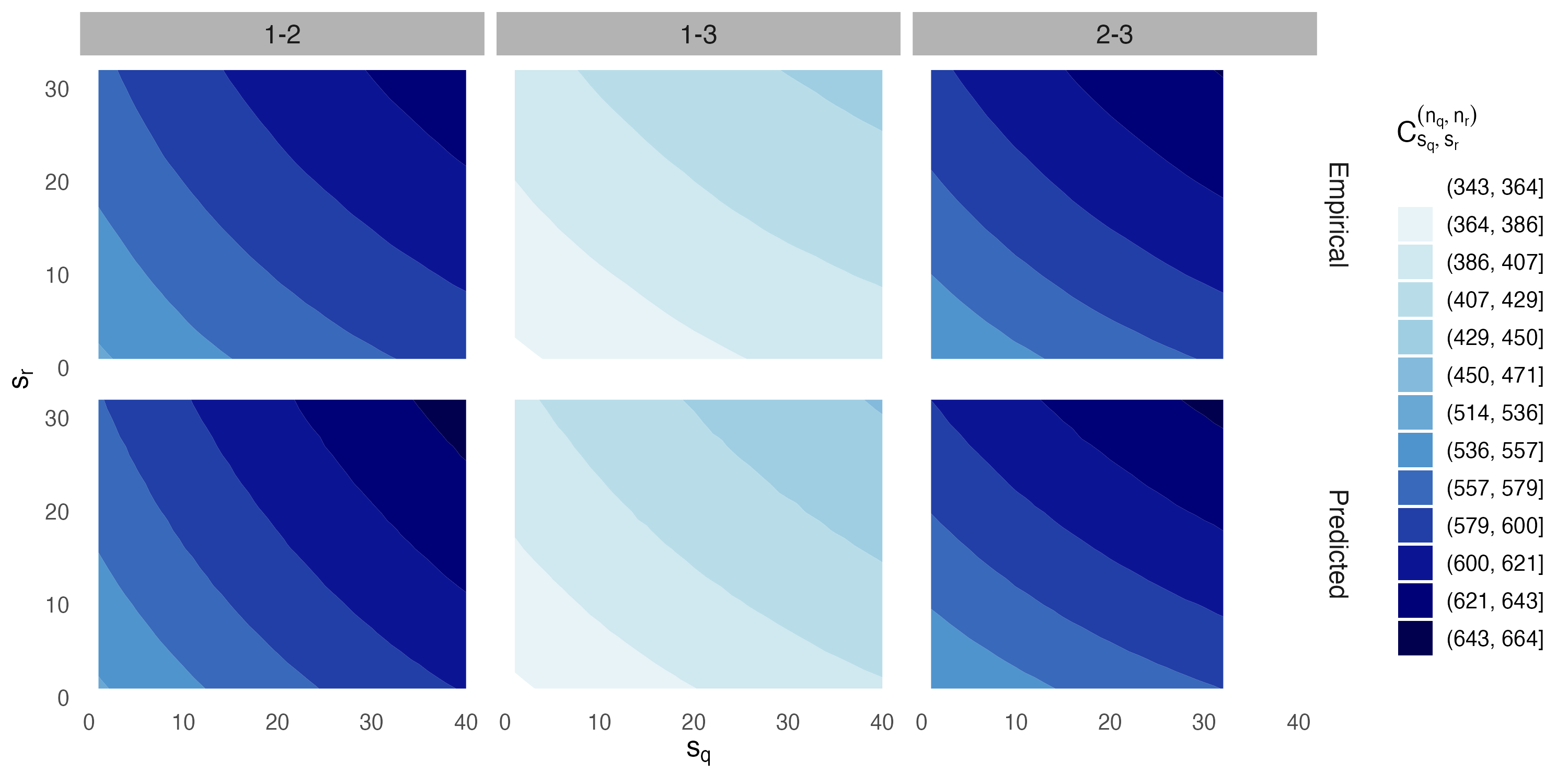}
    \includegraphics[width=0.5\textwidth]{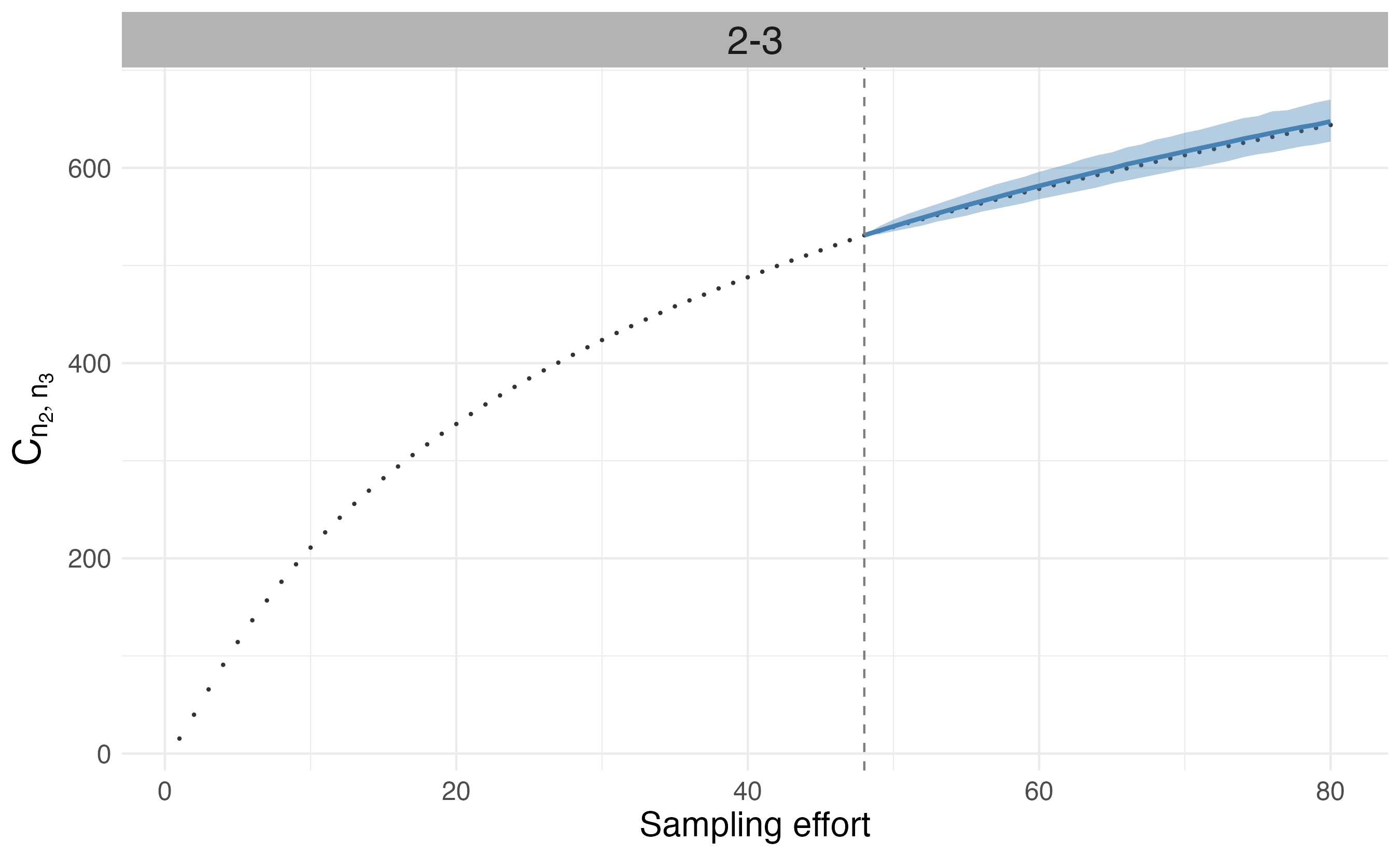}
    \caption{Top panel: empirical rarefaction surfaces for shared species in the test set, together with the corresponding posterior mean estimates of $C_{s_q,s_r}^{(n_q,n_r)}$ for all pairs of sites. Bottom panel: simplified one-dimensional rarefaction curve for shared species (black dots), alongside the posterior mean of $C_{s_2,s_3}^{(n_2,n_3)}$ (solid line) and its $95\%$ credible intervals for site pair 2--3. The gray vertical line indicates the training set.}
    \label{fig:sim_shared}
\end{figure}

We next consider a more complex setting with $Q=15$ sites, where the sample sizes $n_q$ are drawn uniformly between 15 and 40 and $N$ follows a negative binomial distribution with mean $5{,}000$ and size $1{,}000$. The parameters $\bm{\zeta}$ and $\bm{\phi}$ are set as in the previous experiment. This configuration reflects typical biodiversity data.
Data are generated under two scenarios: (A) the correctly specified \texttt{MOSAIC} model and (B) a misspecified version in which a random intercept is added to the linear predictor in~\eqref{eq:linpred_h}, $\log\{\mu_q/(1-\mu_q)\} = \bm{z}_q^{\top}\bm{\zeta} + u_q$, with $u_q \sim \mathcal{N}(0,\tau^2)$ and $\tau = 0.1$. For each scenario, we replicate the experiment 25 times.
Each dataset is split into training and test sets, with $10$ observations in the test set for all sites.
The \texttt{MOSAIC} model is fitted to the training data using the same hyperparameter configuration as in the previous experiment. Based on the fitted model, we predict the number of new species at each site, $K_{q,s_q}^{(n_q)}$, with $s_q=10$ for $q=1,\dots,Q$. Predictive performance is evaluated using the average mean squared error across sites, scaled by the true values, for each $s_q=1,\dots,10$. We compare our approach with the exchangeable negative binomial mixture of beta--Bernoulli models \citep{ghilFeature}, using the same negative binomial hyperparameter specification as in our model.
Figure \ref{fig:sim_MSE} shows that the proposed method achieves comparable or improved predictive performance in both scenarios, highlighting the benefit of borrowing information across sites.
The benchmark is fitted to site-specific counts and does not provide a joint model for species sharing across sites. In the Supplementary Material, we provide additional results for experiments with varying numbers of sites $Q$, together with additional experiments evaluating the frequentist coverage of posterior credible intervals for the regression coefficients $\bm{\zeta}$ and precision parameters $\bm{\phi}$.

\begin{figure}[t!]
    \centering
    \includegraphics[width=0.85\textwidth]{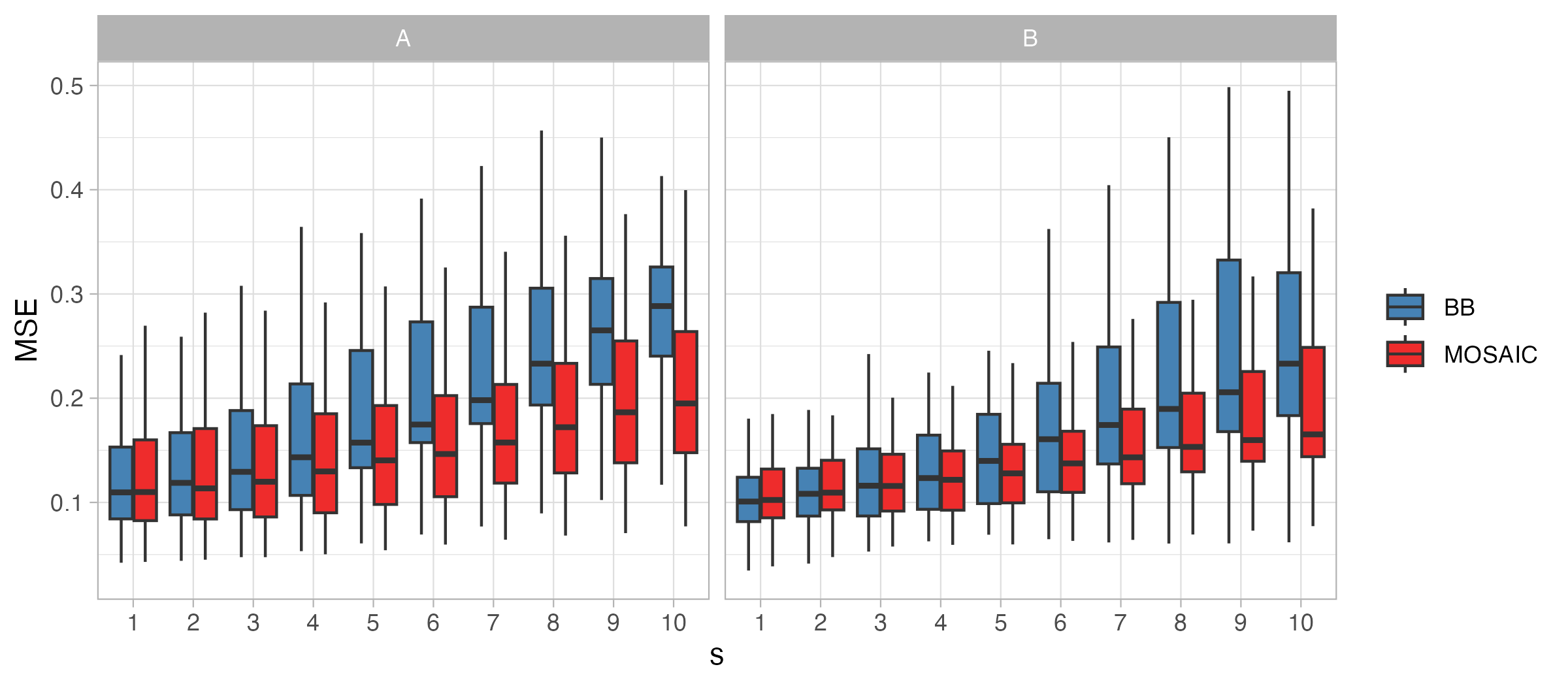}
    \caption{Average MSE in predicting $K_{q,s_q}^{(n_q)}$ for each $s_q=1,\dots,10$ for all $q=1,\dots,Q$ for \texttt{MOSAIC} and the exchangeable negative binomial mixture of beta-Bernoulli (BB).}
    \label{fig:sim_MSE}
\end{figure}

\section{Global patterns in airborne fungal diversity} \label{sec:application}

We analyze data from a global project on the biodiversity of fungi, the GSSP \citep{ovaskainen2024global, abrego2024airborne}. Fungi are among the most diverse and ecologically important kingdoms of life, despite fungal diversity and its environmental drivers remaining poorly understood \citep{abrego2024airborne}.
We consider samples in temperate and polar-continental climatic zones, for a total of $Q=34$ sites in the Northern hemisphere, covering North America, Europe, Asia and the Arctic. In the Supplementary Material we show the map with the sampling locations. The data consist of $1{,}425$ observations over 1 year with site-specific sample sizes $n_q$ ranging between $12$ and $103$, with a median of $37$. The Supplementary Material reports the GSSP labels used in the following to denote the sampling locations, together with the corresponding sample sizes for all sites.

The total number of identified species is $k = 15{,}352$. Data are very sparse, as more than half of the species ($57\%$) have only been observed once. To account for climatic and latitudinal effects, we include latitude and mean annual temperature along with its quadratic term as covariates. In addition, weather covariates may influence variation in fungal communities \citep{abrego2024airborne}. Consequently, we also include the mean annual wind speed and the mean annual precipitation for each sampling area, leading to $H=6$ columns of $\bm{Z}$.

We fit \texttt{MOSAIC} using the MCMC specification in the simulation study.
The posterior mean of the global number of species $N$ is $16{,}574$, with $95\%$ credible interval $(16{,}494,\, 16{,}659)$. This estimate exceeds the $15{,}352$ observed species, suggesting that about $1{,}200$ species remain undetected.
Table \ref{tab:beta} reports estimates of the regression coefficients $\bm{\zeta}$. The large negative intercept reflects the extreme sparsity of the data.
Temperature and latitude show strong associations with recorded fungal occurrence, consistent with \citet{abrego2024airborne} and previous studies on soil fungi \citep{tedersoo2014global}. Estimated occurrence probabilities decrease with distance from the Equator.
The temperature coefficients imply a hump-shaped relationship on the logit scale, consistent with the thermal performance curves discussed by \citet{hess2025evolution}, although these associations do not establish a physiological mechanism.
The negative sign of the mean annual precipitation coefficient is consistent with rainfall washout, where fungal spores are scavenged from the air and deposited onto surfaces. Several studies have reported negative relationships between rainfall and airborne fungal abundance, see, e.g., \citet{pakpour2015relationships}, with genera such as \emph{Cladosporium} and \emph{Alternaria} being particularly susceptible to rainfall-induced declines in airborne occurrence \citep{ho2005characteristics}.

\begin{table}[t]
\centering
\caption{Posterior summaries for $\bm{\zeta}$.}
\label{tab:beta}
\vspace{6pt}
\begin{tabular}{lrc}
\hline
Covariate & \multicolumn{1}{c}{Posterior mean} & \multicolumn{1}{c}{95\% credible interval} \\
\hline
Intercept       & $-5.55$ & $(-5.56,\,-5.53)$ \\
Latitude        & $-0.27$ & $(-0.30,\,-0.25)$ \\
Temperature     & $ 0.67$ & $( 0.64,\, 0.71)$ \\
Temperature$^2$ & $-0.47$ & $(-0.50,\,-0.44)$ \\
Wind            & $-0.11$ & $(-0.14,\,-0.09)$ \\
Precipitation   & $-0.26$ & $(-0.28,\,-0.24)$ \\
\hline
\end{tabular}
\end{table}

To assess diversity in fungal composition we compute the $\beta$-diversity index proposed in Section \ref{sec:betadiv}. Figure \ref{fig:fungi_beta} displays the posterior mean of the similarity $1-\beta_{qr}$ for all pairs of sites, with detectability parameter $\sigma_q=0.6$ for all sites. The sensitivity analysis in Figure \ref{fig:sm_betaSigma} of the Supplementary Material shows similar overall patterns for the detectability scenarios considered. The results are compared to those obtained via the Jaccard and Ochiai similarity indices \citep{Ochiai1957}. Although the overall patterns are broadly consistent, the proposed index reveals more complex and nuanced biodiversity patterns, particularly for sites with small sample sizes. For example, our measure identifies a small cluster of High Arctic and Arctic tundra locations, comprising CHA, ZAC, NUU, PRI and TOO, that is not visible in the Jaccard index. These sites are associated with cold-climate ecosystems and are geographically isolated, all having relatively small sample sizes $n_q$, highlighting the importance of borrowing information across locations in the estimation procedure.

In general, airborne fungal communities exhibit a pronounced biogeographic organization, with temperate European sites forming a highly connected cluster (from KLA to ZUR) and polar-continental locations in Europe grouping separately (from TUR to TAR). These patterns are consistent with evidence that fungal community composition is strongly shaped by climatic zone \citep{abrego2024airborne}. Shared climate and atmospheric conditions may contribute to the similarity among temperate European sites, while environmental filtering and geographic isolation may contribute to the distinct Arctic and polar-continental patterns. For example, the marked dissimilarity of CAM (Sandger\o{}i, Iceland) and CHA (Cambridge Bay, Canada) from most other locations is ecologically meaningful and consistent with Arctic aerobiology and microbial biogeography theory \citep{cao2024fungal}. Possible explanations include their atypical climatic settings, with Sandger\o{}i influenced by strong wind exposure and North Atlantic air-mass transport and Cambridge Bay characterized by extreme continental isolation and sparse vegetation, both of which may shape local airborne fungal assemblages differently from typical polar-continental sites.

\begin{figure}[t]
    \centering
    \includegraphics[width=0.99\textwidth]{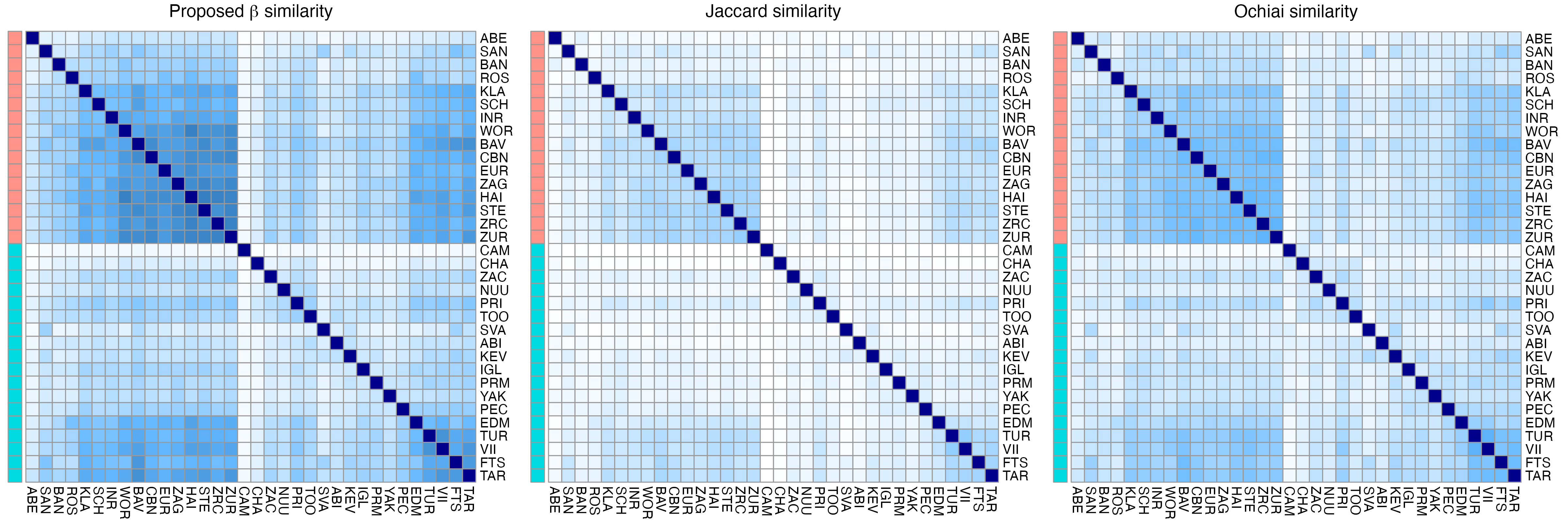}
    \caption{Posterior mean of the proposed $\beta$-similarity for all sites, grouped by climatic zone, compared with Jaccard and Ochiai similarities.}
    \label{fig:fungi_beta}
\end{figure}

To further explore fungal biodiversity, we produce extrapolations of rarefaction curves for shared species between pairs of sites, $C_{s_q,s_r}^{(n_q,n_r)}$. Since most sites have different sample sizes, the one-dimensional summary curve taking the diagonal of the whole surface has limited utility here. We instead construct extrapolations of rarefaction curves by simulating from the posterior under two prediction sets: (i) a one-side test set, where we simulate additional observations only for the site with the smaller sample size and compute the number of shared species using the remaining observations from the other site; and (ii) a full test set, where we simulate additional observations from both sites. Assuming $n_q > n_r$, the size of the one-side test set is $n_q - n_r$. This approach makes full use of the available data, avoiding unnecessary discarding of observations.
The results with test size $s_q = 1,\dots,30$ are shown in Figure \ref{fig:fungi_semishared} for three illustrative site pairs: a pair of temperate locations, Klagenfurt (KLA) and Zurich (ZUR); a pair of polar-continental locations, Nuuk (NUU) and Wrangel Island (PRI); and a cross-zone pair, Klagenfurt and Fetsund (FTS). The figure also displays the posterior density of $\beta_{qr}$ for each pair.
The expected number of new shared species grows rather slowly when collecting additional balanced samples for the Nuuk--Wrangel Island pair, which exhibits the highest $\beta$-diversity among the three.
As previously noted, Arctic locations tend to appear quite distinct, likely reflecting the combined effects of geographic isolation, sparse tundra vegetation, and limited regional spore sources. Growth is faster for the remaining two pairs, which comprise locations in Europe, even though Fetsund belongs to the polar-continental climatic zone, consistent with lower $\beta$-diversity. In the Supplementary Material we provide additional experiments assessing the model performance in predicting the number of shared species in out-of-sample experiments.

\begin{figure}[t!]
    \centering
    \includegraphics[width=0.92\textwidth]{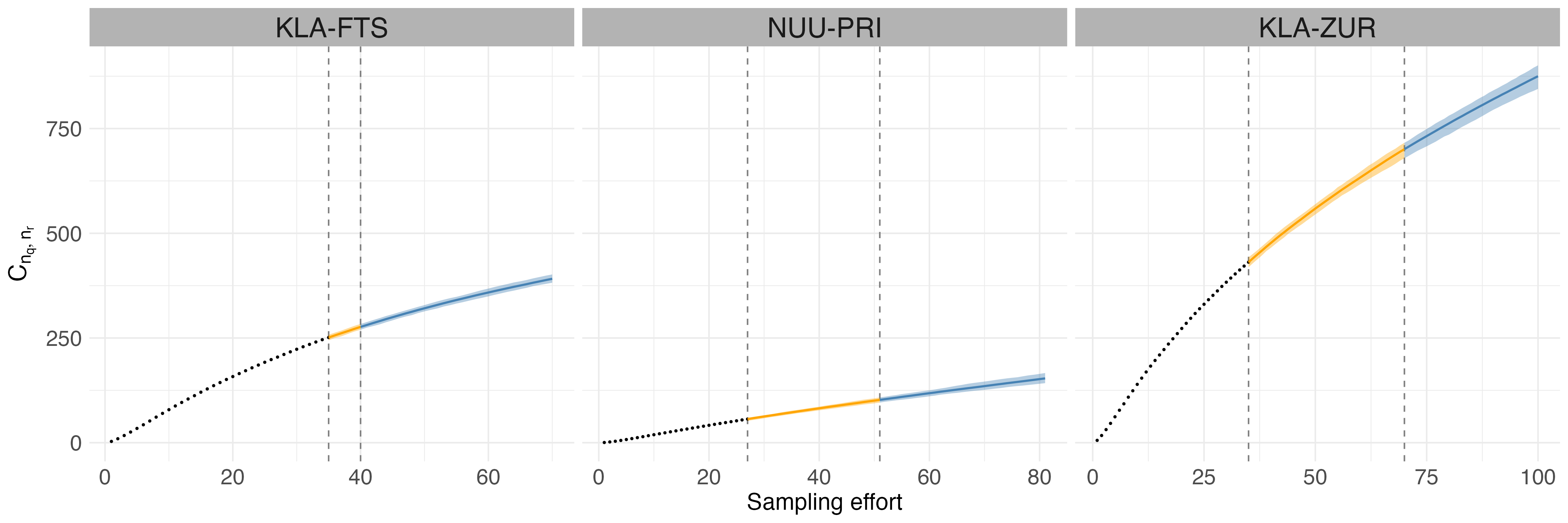}
    \includegraphics[width=0.92\textwidth]{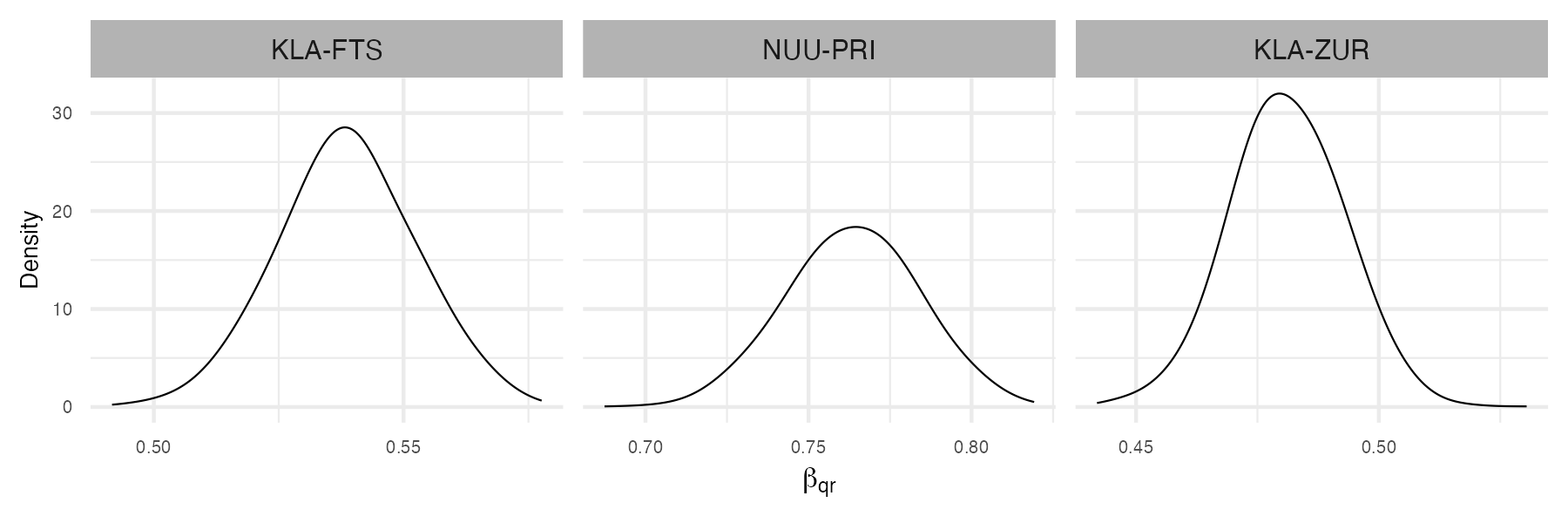}
    \caption{Top panel: empirical shared-species rarefaction curves (dots), posterior mean (solid lines) and 95\% credible intervals (shaded ribbons) for one-side test and test predictions across three sampling location pairs. Gray vertical lines indicate training and one-side test sets. Bottom panel: posterior densities of $\beta_{qr}$ for the same three pairs.}
    \label{fig:fungi_semishared}
\end{figure}

To explore sensitivity to detectability assumptions, Figure~\ref{fig:fungi_latent_rar} shows rarefaction curves for the latent species count $L_{q,n_q}$, the counterpart of $K_{q,n_q}$ based on latent occupancies $\bm{W}$. We consider $\sigma=0.05$, $0.5$ and $0.95$, corresponding to low, intermediate and near-perfect detectability. Expected latent richness increases as assumed detectability decreases. These are sensitivity scenarios: $\sigma$ is not identified from the data, as discussed in Remark~\ref{rem:identifiability}. Section~\ref{sec:supp_TheLatent} of the Supplementary Material gives the derivations and analogous results for latent shared-species counts.

\begin{figure}[t!]
    \centering
    \includegraphics[width=0.99\textwidth]{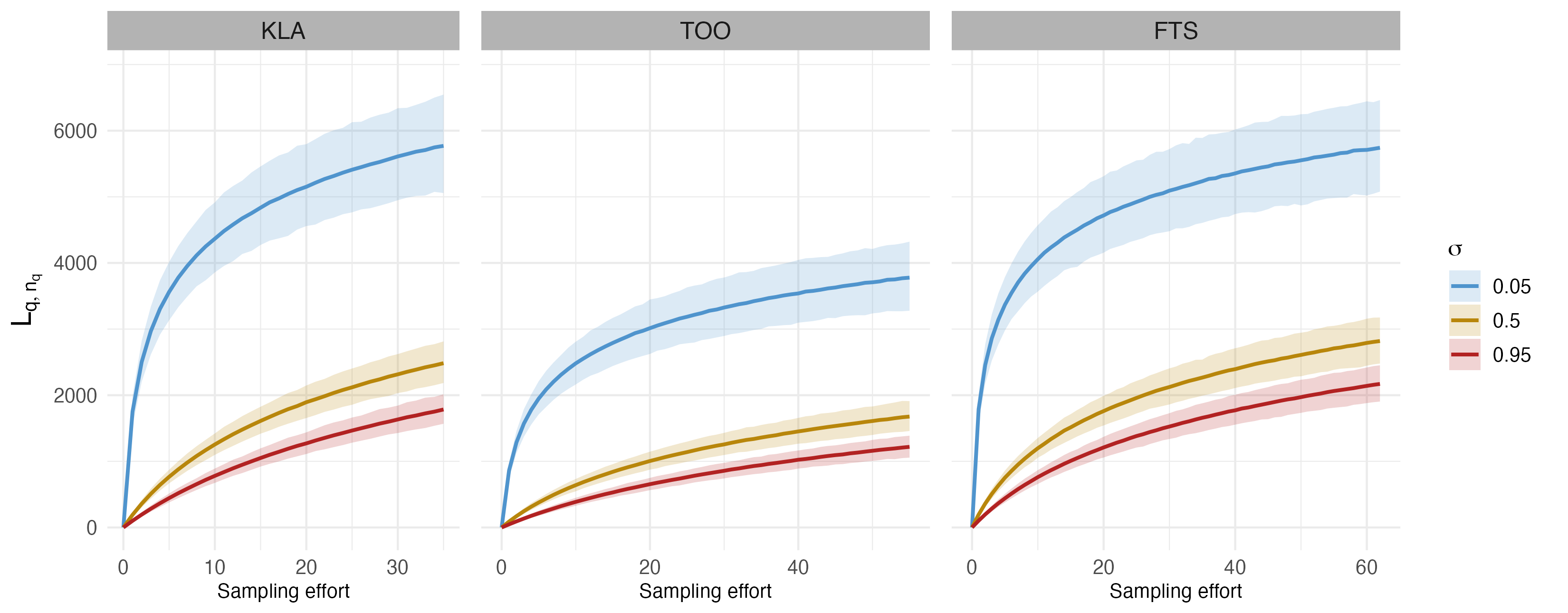}
    \caption{Posterior means and 95\% credible intervals for site-specific rarefaction curves for the latent number of species $L_{q,n_q}$ (latent occupancies) for different values of $\sigma$.}
    \label{fig:fungi_latent_rar}
\end{figure}

\section{Discussion}

We have proposed \texttt{MOSAIC}, a novel partially exchangeable feature allocation framework for inference on species sharing and $\beta$-diversity under imperfect detection. The construction balances ecological realism with analytical tractability and computational efficiency, enabling inference in large biodiversity surveys.

To obtain analytic tractability, we have made some simplifying assumptions that are not entirely realistic in ecology. For example, unlike joint species distribution models (\texttt{JSDMs}), we have not modeled statistical dependence in co-occurrence across species. Modern versions of \texttt{JSDMs} typically rely on Bayesian factor models for high-dimensional multivariate binary outcome species occurrence data, and substantial challenges arise in implementing such models for highly diverse communities containing 1,000s to millions of species \citep{mauri2025factor}. In addition, a key aspect that forms much of the focus of the \texttt{MOSAIC} framework is the ability to allow and predict new species discovery as sampling proceeds, which is currently lacking in \texttt{JSDMs} with the exception of \cite{stolf2025infinite}. In future work, it will be of substantial interest to define new approaches bridging between \texttt{JSDMs} and feature allocation models such as \texttt{MOSAIC}.

A further challenge is the non-identifiability of occupancy and detectability. Other ecological models use targeted sampling designs, such as capture--recapture, to estimate detection probabilities \citep{pollock1990statistical}. Such designs do not transfer directly to destructive sampling and DNA meta-barcoding. Developing suitable replication or calibration designs to inform detectability is an important direction for future work.



\section*{Acknowledgement}
This research was partially supported by the European Research Council under the European Union’s Horizon 2020 research and innovation programme (grant agreement No 856506), the National Science Foundation (IIS-2426762) and the Office of Naval Research (grant N000142412626). The authors thank Otso Ovaskainen, Tomas Roslin, and other members of the Lifeplan project team for helpful discussions.



\bibliography{paper-ref}

\newpage

\section*{Supplementary material}

\setcounter{section}{0}
\setcounter{table}{0}
\setcounter{figure}{0}
\setcounter{equation}{0}
\setcounter{proposition}{0}
\setcounter{theorem}{0}
\setcounter{lemma}{0}
\renewcommand{\thelemma}{S\arabic{lemma}}
\renewcommand{\thetheorem}{S\arabic{theorem}}
\renewcommand{\theproposition}{S\arabic{proposition}}
\renewcommand{\thefigure}{S\arabic{figure}}
\renewcommand{\thesection}{S.\arabic{section}}
\renewcommand{\thetable}{S\arabic{table}}
\renewcommand{\theequation}{S\arabic{equation}}
\def\theHsection{S.\arabic{section}}
\def\theHfigure{S\arabic{figure}}
\def\theHtable{S\arabic{table}}
\def\theHequation{S\arabic{equation}}
\def\theHlemma{S\arabic{lemma}}
\def\theHtheorem{S\arabic{theorem}}
\def\theHproposition{S\arabic{proposition}}

\section{Proofs of theoretical results}

\subsection{Background results}

We characterize the posterior distributions of the occurrence probabilities $\pi_{jq} = \eta_{jq} \gamma_{jq}$ and of the total number of species $N$, both of which admit simple closed-form expressions. Let $N' = N - k$ denote the number of unseen species given $\bm{Y}$, and split the occurrence probabilities into two groups: those associated with the observed species $X_1,\dots, X_k$, denoted by $\pi^*_{jq}$ for $j=1,\dots,k$, and those associated with the unobserved ones, denoted by $\pi'_{jq}$ for $j=1,\dots,N'$, in both cases with $q=1,\dots,Q$.

\begin{proposition} \label{th:posterior_known}
Let $\bm{Y}$ follow the \texttt{MOSAIC} model with parameters $(\lambda, \bm{\alpha}, \bm{\theta}, \bm{\sigma})$ and set $\pi_{jq} = \eta_{jq}\gamma_{jq}$. Then the posterior distribution of the number of unseen species is
\begin{equation*}
    N'\mid \bm{Y} \sim \mathrm{Poisson}\{\lambda \, p_0(\bm{\alpha}, \bm{\theta}, \bm{n})\}, \qquad p_0(\bm{\alpha}, \bm{\theta}, \bm{n}) = \prod_{q=1}^Q p_{0,q}(\alpha_q, \theta_q, n_q),
\end{equation*}
with $p_{0,q}(\alpha_q, \theta_q, n_q) = (\alpha_q + \theta_q)_{n_q}/(\theta_q)_{n_q}$. Moreover, the occurrence probabilities are independent across species and sites given $\bm{Y}$, with
\begin{equation*}
\begin{aligned}
\pi_{jq}^* \mid \bm{Y} &\overset{\textup{ind}}{\sim} \mathrm{Beta}(m_{jq}-\alpha_q,\ \alpha_q+\theta_q+n_q-m_{jq}), &&\quad (j=1,\dots,k),\\
\pi_{jq}' \mid \bm{Y}  &\overset{\textup{ind}}{\sim} \mathrm{Beta}(-\alpha_q,\ \alpha_q+\theta_q+n_q), &&\quad (j=1,\dots,N'),
\end{aligned}
\end{equation*}
for $q=1,\dots,Q$. In particular, the posterior does not depend on $\bm{\sigma}$.
\end{proposition}

Proposition \ref{th:posterior_known} is a special case of a general theorem in \cite{ghilotti2025bayesian}.
We note that $p_{0,q}(\alpha_q, \theta_q, n_q) = (\alpha_q + \theta_q)_{n_q}/(\theta_q)_{n_q} = E\{(1 - \pi_{jq})^{n_q}\}$ is the probability that a species goes unrecorded in the $n_q$ samples from site $q$, and that $p_0(\bm{\alpha}, \bm{\theta}, \bm{n}) = \prod_{q=1}^Q p_{0,q}(\alpha_q, \theta_q, n_q)$ is the probability that it goes unrecorded at every site.
Moreover, the posterior expectations of the occurrence probabilities are $E(\pi^*_{jq} \mid \bm{Y}) = (m_{jq}-\alpha_q)/(\theta_q+n_q)$ and $E(\pi'_{jq} \mid \bm{Y}) = -\alpha_q / (\theta_q+n_q)$, for $j=1,\dots,k$ and $j=1,\dots,N'$ respectively, and $q=1,\dots,Q$. The first is a shrunken version of the empirical occurrence frequency $m_{jq}/n_q$, with $-\alpha_q$ acting as a pseudo-count and $\theta_q$ as a pseudo-sample size.

\subsection{Auxiliary results}

\begin{lemma} \label{lemma:sm_pgf}
Let $V \ge 1$ be an integer and let $\varrho_1,\dots,\varrho_V$ be independent and identically distributed random variables with values in $(0,1)$. Let $G_1,\dots,G_V$ be conditionally independent Bernoulli random variables given $(\varrho_1,\dots,\varrho_V)$, with $\pr(G_j = 1 \mid \varrho_1,\dots,\varrho_V) = \varrho_j$ for $j=1,\dots,V$. Then
\begin{equation*}
    \sum_{j=1}^{V} G_j \sim \mathrm{Binomial}\{V, E(\varrho_1)\}.
\end{equation*}
\end{lemma}
\begin{proof}
Consider the probability generating function $G(z) = E(z^{\sum_{j=1}^V G_j})$. By the tower property and the conditional independence of the $G_j$'s,
\begin{equation*}
    E\Big(z^{\sum_{j=1}^V G_j} \Big)
    = E\Big\{ E\Big(z^{\sum_{j=1}^V G_j} \mid \varrho_1,\dots,\varrho_V \Big) \Big\}
    = E\Big\{ \prod_{j=1}^{V} E\big( z^{G_j} \mid \varrho_1,\dots,\varrho_V \big) \Big\}.
\end{equation*}
Since $E( z^{G_j} \mid \varrho_1,\dots,\varrho_V) = 1 - \varrho_j + \varrho_j z$ depends on $\varrho_j$ only, and the $\varrho_j$'s are independent and identically distributed,
\begin{equation*}
    E\Big(z^{\sum_{j=1}^V G_j} \Big)
    = \prod_{j=1}^{V} E( 1 - \varrho_j + \varrho_j z )
    = \{ 1 - E(\varrho_1) + E(\varrho_1) z \}^{V},
\end{equation*}
which is the probability generating function of a $\mathrm{Binomial}\{V,E(\varrho_1)\}$ random variable.
\end{proof}

\begin{lemma} \label{lemma:sm_thinning}
If $V \sim \mathrm{Poisson}(\mu)$ and $X \mid V \sim \mathrm{Binomial}(V,p)$, then
$X \sim \mathrm{Poisson}(\mu p)$.
\end{lemma}
\begin{proof}
For any integer $x \ge 0$,
\begin{align*}
    \pr(X = x)
    &= \sum_{v \ge x} \binom{v}{x} p^x (1-p)^{v-x} \frac{\mu^v e^{-\mu}}{v!}
    = \frac{(\mu p)^x}{x!} e^{-\mu} \sum_{v \ge x} \frac{\{\mu(1-p)\}^{v-x}}{(v-x)!}\\
    &= \frac{(\mu p)^x}{x!} e^{-\mu} e^{\mu(1-p)}
    = \frac{(\mu p)^x}{x!} e^{-\mu p},
\end{align*}
which is the probability mass function of a Poisson distribution with mean $\mu p$.
\end{proof}

\subsection{Proof of Proposition~\ref{pr:Kprior}}
Conditionally on $N$, the number of distinct species observed in the $n_q$ samples collected at site $q$ is
\begin{equation*}
    K_{q,n_q} = \sum_{j=1}^N \mathds{1}\Big( \sum_{i=1}^{n_q} \tilde{Y}_{ijq}>0 \Big) = \sum_{j=1}^N G_{jq},
\end{equation*}
where the $G_{jq}$ are independent Bernoulli given the random probabilities
\begin{equation*}
    \varrho_{jq} := \pr(G_{jq}=1 \mid \pi_{jq})
    = 1- \pr\Big( \sum_{i=1}^{n_q} \tilde{Y}_{ijq}=0 \mid \pi_{jq}\Big)
    = 1- (1-\pi_{jq})^{n_q}.
\end{equation*}
The $\pi_{jq}$'s are independent and identically distributed across $j$, and hence so are the $\varrho_{jq}$'s, being the same measurable transformation of them. Lemma \ref{lemma:sm_pgf} then gives $K_{q,n_q} \mid N \sim \mathrm{Binomial}\{N,E(\varrho_{1q})\}$ with
\begin{equation*}
    E(\varrho_{1q}) = 1 - E\{(1-\pi_{1q})^{n_q}\}
      = 1 - p_{0,q}(\alpha_q,\theta_q,n_q)
      = p_{1,q}(\alpha_q,\theta_q,n_q).
\end{equation*}
Since $N \sim \mathrm{Poisson}(\lambda)$, Lemma \ref{lemma:sm_thinning} yields $K_{q,n_q} \sim \mathrm{Poisson}\{\lambda\, p_{1,q}(\alpha_q,\theta_q,n_q)\}$.
Conditionally on $N$, the total number of distinct species observed in $\bm{n}$ samples is
\begin{equation*}
    K_{\bm{n}} = \sum_{j=1}^N \mathds{1}\Big(\sum_{q=1}^Q \sum_{i=1}^{n_q} \tilde{Y}_{ijq}>0\Big) = \sum_{j=1}^N G_j,
\end{equation*}
where the $G_j$ are independent Bernoulli given the random probabilities
\begin{equation*}
    \varrho_{j} := \pr(G_j=1 \mid \pi_{j1}, \dots, \pi_{jQ})
    = 1- \pr\Big(\sum_{q=1}^Q \sum_{i=1}^{n_q} \tilde{Y}_{ijq}=0 \mid \pi_{j1},\dots,\pi_{jQ}\Big)
    = 1- \prod_{q=1}^Q(1-\pi_{jq})^{n_q}.
\end{equation*}
The vectors $(\pi_{j1},\dots,\pi_{jQ})$ are independent and identically distributed across $j$, and hence so are the $\varrho_j$'s. Lemma \ref{lemma:sm_pgf} then gives $K_{\bm{n}} \mid N \sim \mathrm{Binomial}\{N,E(\varrho_1)\}$ with
\begin{equation*}
    E(\varrho_1) = 1- E\Big\{ \prod_{q=1}^Q(1-\pi_{1q})^{n_q} \Big\}
      = 1- \prod_{q=1}^Q E\{(1-\pi_{1q})^{n_q} \}
      = 1- \prod_{q=1}^Q p_{0,q}(\alpha_q,\theta_q,n_q)
      = p_1(\bm{\alpha},\bm{\theta},\bm{n}),
\end{equation*}
where the second equality uses the independence of the occurrence probabilities across sites. Lemma \ref{lemma:sm_thinning} then gives $K_{\bm{n}} \sim \mathrm{Poisson}\{\lambda\, p_1(\bm{\alpha},\bm{\theta},\bm{n})\}$.

\subsection{Proof of Proposition~\ref{pr:shared}}
Conditionally on $N$, the number of species shared between two sites $q$ and $r$, with $q \neq r$, is
\begin{equation*}
    C_{n_q,n_r} = \sum_{j=1}^N G_{j,qr},
    \qquad
    G_{j,qr} = \mathds{1}\Big(\sum_{i=1}^{n_q} \tilde{Y}_{ijq}>0\Big)\,\mathds{1}\Big(\sum_{i=1}^{n_r} \tilde{Y}_{ijr}>0\Big),
\end{equation*}
where the $G_{j,qr}$ are independent Bernoulli given the random probabilities
\begin{align*}
     \varrho_{j,qr} := \pr(G_{j,qr}=1 \mid \pi_{jq}, \pi_{jr})
     &= \Big\{ 1-\pr\Big(\sum_{i=1}^{n_q} \tilde{Y}_{ijq}=0 \mid \pi_{jq}\Big)\Big\}
        \Big\{ 1-\pr\Big(\sum_{i=1}^{n_r} \tilde{Y}_{ijr}=0 \mid \pi_{jr}\Big)\Big\}\\
     &= \{1-(1-\pi_{jq})^{n_q}\} \{1-(1-\pi_{jr})^{n_r}\},
\end{align*}
the factorization following from the conditional independence of the samples collected at the two sites. The pairs $(\pi_{jq},\pi_{jr})$ are independent and identically distributed across $j$, and hence so are the $\varrho_{j,qr}$'s, so Lemma \ref{lemma:sm_pgf} gives $C_{n_q,n_r} \mid N \sim \mathrm{Binomial}\{N,E(\varrho_{1,qr})\}$ with
\begin{align*}
    E(\varrho_{1,qr}) &= E\big[\{1-(1-\pi_{1q})^{n_q}\} \{1-(1-\pi_{1r})^{n_r}\}\big]
       = \big[1-E\{(1-\pi_{1q})^{n_q}\}\big]\big[1-E\{(1-\pi_{1r})^{n_r}\}\big]\\
       &= p_{1,q}(\alpha_q,\theta_q,n_q)\, p_{1,r}(\alpha_r,\theta_r,n_r),
\end{align*}
where the second equality uses the independence of the occurrence probabilities across sites. Since $N \sim \mathrm{Poisson}(\lambda)$, Lemma \ref{lemma:sm_thinning} then yields $C_{n_q,n_r} \sim \mathrm{Poisson}\{\lambda\, p_{1,q}(\alpha_q,\theta_q,n_q)\, p_{1,r}(\alpha_r,\theta_r,n_r)\}$.

\subsection{Proof of Theorem~\ref{thm:Kpred}}

Let $\mathcal{A} = \{ j \in \{1,\dots,N\} : \sum_{q=1}^Q \sum_{i=1}^{n_q} \tilde{Y}_{ijq}=0 \}$ be the set of species unrecorded at all sites, so that $|\mathcal{A}| = N-k$ given $\bm{Y}$. In the \texttt{MOSAIC} model the number of new species that would be discovered in $\bm{s}$ additional samples is
\begin{equation*}
    K_{\bm{s}}^{(\bm{n})} = \sum_{j \in \mathcal{A}}
    \mathds{1}\Big(\sum_{q=1}^Q\sum_{i=n_q+1}^{n_q+s_q} \tilde{Y}_{ijq}>0\Big) = \sum_{j \in \mathcal{A}} G_j,
\end{equation*}
where the $G_j$'s are independent Bernoulli given the random probabilities
\begin{equation*}
    \varrho_j := \pr\Big(\sum_{q=1}^Q\sum_{i=n_q+1}^{n_q+s_q} \tilde{Y}_{ijq}>0 \mid \pi_{j1},\dots,\pi_{jQ}\Big)
    = 1- \prod_{q=1}^Q (1-\pi_{jq})^{s_q}.
\end{equation*}
By Proposition~\ref{th:posterior_known}, the occurrence probabilities of the previously unobserved species satisfy $\pi'_{jq} \mid \bm{Y} \overset{\textup{ind}}{\sim} \mathrm{Beta}(-\alpha_q, \alpha_q+\theta_q+n_q)$, independently across species and sites, so that the $\varrho_j$'s are independent and identically distributed across $j \in \mathcal{A}$ given $\bm{Y}$, with
\begin{equation*}
    E(\varrho_j \mid \bm{Y})
     = 1- \prod_{q=1}^Q  E\{(1-\pi_{jq})^{s_q} \mid \bm{Y}\}
       = 1 - \prod_{q=1}^Q p_{0,q}(\alpha_q,\theta_q+n_q,s_q)
       = p_1(\bm{\alpha},\bm{\theta}+\bm{n},\bm{s}).
\end{equation*}
Lemma \ref{lemma:sm_pgf} then gives $K_{\bm{s}}^{(\bm{n})} \mid \bm{Y}, N \sim \mathrm{Binomial}\{N-k,\, p_1(\bm{\alpha},\bm{\theta}+\bm{n},\bm{s})\}$. By the same proposition, $N-k \mid \bm{Y} \sim \mathrm{Poisson}\{\lambda\, p_0(\bm{\alpha},\bm{\theta},\bm{n})\}$, so that marginalizing through Lemma \ref{lemma:sm_thinning} gives $K_{\bm{s}}^{(\bm{n})} \mid \bm{Y} \sim \mathrm{Poisson}\{\lambda\, p_0(\bm{\alpha},\bm{\theta},\bm{n})\, p_1(\bm{\alpha},\bm{\theta}+\bm{n},\bm{s})\}$.

Let now $\mathcal{A}_q = \{ j \in \{1,\dots,N\} : \sum_{i=1}^{n_q} \tilde{Y}_{ijq}=0 \}$ be the set of species not yet recorded at site $q$, so that $|\mathcal{A}_q| = N-k_q$; note that $\mathcal{A}_q$ contains species that have been recorded at some other site. The number of species new to site $q$ that would be discovered in $s_q$ additional samples is
\begin{equation*}
    K_{q,s_q}^{(n_q)} = \sum_{j \in \mathcal{A}_q} \mathds{1}\Big(\sum_{i=n_q+1}^{n_q+s_q} \tilde{Y}_{ijq}>0\Big)
    = \sum_{j \in \mathcal{A}_q} G_{jq},
\end{equation*}
where the $G_{jq}$'s are independent Bernoulli given the random probabilities $\varrho_{jq} := 1-(1-\pi_{jq})^{s_q}$, which are independent and identically distributed across $j \in \mathcal{A}_q$ given $\bm{Y}$, again by Proposition~\ref{th:posterior_known}, and satisfy
\begin{equation*}
    E(\varrho_{jq} \mid \bm{Y})
     = 1 - E\{(1-\pi_{jq})^{s_q} \mid \bm{Y}\}
     = 1- p_{0,q}(\alpha_q,\theta_q+n_q,s_q)
     = p_{1,q}(\alpha_q,\theta_q+n_q,s_q).
\end{equation*}
We decompose $\mathcal{A}_q = \mathcal{A} \cup (\mathcal{A}_q \setminus \mathcal{A})$, where $\mathcal{A}_q \setminus \mathcal{A}$ is the set of species recorded at some site other than $q$, whose cardinality $k-k_q$ is known given $\bm{Y}$. Applying Lemma \ref{lemma:sm_pgf} to each of the two groups separately gives
\begin{align*}
    K_{q,s_q}^{(n_q)} \mid \bm{Y}, N \overset{d}{=}\;
    &\mathrm{Binomial}\{N-k,\, p_{1,q}(\alpha_q,\theta_q+n_q,s_q)\}\\
    +\;&\mathrm{Binomial}\{k-k_q,\, p_{1,q}(\alpha_q,\theta_q+n_q,s_q)\},
\end{align*}
the two terms being independent, since they involve disjoint sets of species with independent occurrence probabilities. Marginalizing the first term over $N-k \mid \bm{Y} \sim \mathrm{Poisson}\{\lambda\, p_0(\bm{\alpha},\bm{\theta},\bm{n})\}$ through Lemma \ref{lemma:sm_thinning}, we obtain
\begin{align*}
    K_{q,s_q}^{(n_q)} \mid \bm{Y} \overset{d}{=}\;
    &\mathrm{Poisson}\{\lambda\, p_0(\bm{\alpha},\bm{\theta},\bm{n})\, p_{1,q}(\alpha_q,\theta_q+n_q,s_q)\}\\
    +\;&\mathrm{Binomial}\{k-k_q,\, p_{1,q}(\alpha_q,\theta_q+n_q,s_q)\}.
\end{align*}

\subsection{Proof of Theorem~\ref{thm:pred_shared}}

Let $q \neq r$ and define the three disjoint sets of species
\begin{equation*}
\begin{aligned}
    \mathcal{A}_{00} &= \Big\{j\in \{1,\dots,N\} : \sum_{i=1}^{n_q}\tilde{Y}_{ijq}=0,\; \sum_{i=1}^{n_r}\tilde{Y}_{ijr}=0\Big\},\\
    \mathcal{A}_{q} &= \Big\{j\in \{1,\dots,N\} : \sum_{i=1}^{n_q}\tilde{Y}_{ijq}>0,\; \sum_{i=1}^{n_r}\tilde{Y}_{ijr}=0\Big\},\\
    \mathcal{A}_{r} &= \Big\{j\in \{1,\dots,N\} : \sum_{i=1}^{n_q}\tilde{Y}_{ijq}=0,\; \sum_{i=1}^{n_r}\tilde{Y}_{ijr}>0\Big\},
\end{aligned}
\end{equation*}
namely the species recorded at neither site, those recorded only at site $q$ and those recorded only at site $r$. To indicate the detection of species $j$ at site $h \in \{q,r\}$ in the additional samples we use $I_{jh} = \mathds{1}(\sum_{i=n_h+1}^{n_h+s_h} \tilde{Y}_{ijh}>0)$. In the \texttt{MOSAIC} model the number of new species shared between sites $q$ and $r$ that would be discovered in $(s_q,s_r)$ additional samples is
\begin{equation} \label{eq:sm_target_cpqpred}
     C_{s_q,s_r}^{(n_q,n_r)} = \sum_{j \in \mathcal{A}_{00}} I_{jq}I_{jr}
     + \sum_{j \in \mathcal{A}_{q}} I_{jr}
     + \sum_{j \in \mathcal{A}_{r}} I_{jq}.
\end{equation}
For $h \in \{q,r\}$, the indicators $I_{jh}$ are independent Bernoulli given the random probabilities $\varrho_{jh} := 1-(1-\pi_{jh})^{s_h}$. By Proposition~\ref{th:posterior_known}, the occurrence probabilities of the previously unobserved species satisfy $\pi'_{jh} \mid \bm{Y} \overset{\textup{ind}}{\sim} \mathrm{Beta}(-\alpha_h, \alpha_h+\theta_h+n_h)$, independently across species and sites, so that
\begin{equation*}
     E(\varrho_{jh} \mid \bm{Y})
     = 1- E\{(1-\pi_{jh})^{s_h} \mid \bm{Y}\}
     = 1- p_{0,h}(\alpha_h,\theta_h+n_h,s_h)
     = p_{1,h}(\alpha_h,\theta_h+n_h,s_h),
\end{equation*}
whereas the species in $\mathcal{A}_{00}$ contribute with random success probability $\varrho_{jq}\varrho_{jr}$, whose expectation factorizes by the independence of $\pi_{jq}$ and $\pi_{jr}$,
\begin{equation*}
    E(\varrho_{jq}\varrho_{jr} \mid \bm{Y})
    = p_{1,q}(\alpha_q,\theta_q+n_q,s_q)\, p_{1,r}(\alpha_r,\theta_r+n_r,s_r).
\end{equation*}
Within each of the sets $\mathcal{A}_{00}$, $\mathcal{A}_{q}$ and $\mathcal{A}_{r}$ these success probabilities are independent and identically distributed across $j$ given $\bm{Y}$.

It remains to identify the cardinalities of the three sets. Recalling that $k_q$ and $k_r$ denote the number of species recorded at sites $q$ and $r$, that $c_{q,r}$ is the number of species recorded at both and that $k_{q,r} = k_q+k_r-c_{q,r}$, we have $|\mathcal{A}_{q}| = k_q - c_{q,r}$ and $|\mathcal{A}_{r}| = k_r - c_{q,r}$, both known given $\bm{Y}$. As for $\mathcal{A}_{00}$, we split it into the species that are unrecorded everywhere, which are $N-k$ and hence random, and the species recorded at some site other than $q$ and $r$, which are $k - k_{q,r}$ and hence known given $\bm{Y}$. Applying Lemma \ref{lemma:sm_pgf} to each of the four groups separately, we obtain
\begin{align*}
     C_{s_q,s_r}^{(n_q,n_r)} \mid \bm{Y}, N \overset{d}{=}\;
     &\mathrm{Binomial}\{N-k,\, p_{1,q}(\alpha_q,\theta_q+n_q,s_q)\, p_{1,r}(\alpha_r,\theta_r+n_r,s_r)\}\\
     +\;&\mathrm{Binomial}\{k-k_{q,r},\, p_{1,q}(\alpha_q,\theta_q+n_q,s_q)\, p_{1,r}(\alpha_r,\theta_r+n_r,s_r)\}\\
     +\;&\mathrm{Binomial}\{k_q-c_{q,r},\, p_{1,r}(\alpha_r,\theta_r+n_r,s_r)\}\\
     +\;&\mathrm{Binomial}\{k_r-c_{q,r},\, p_{1,q}(\alpha_q,\theta_q+n_q,s_q)\},
\end{align*}
where the four terms are independent because they involve disjoint sets of species with independent occurrence probabilities. Finally, by Proposition~\ref{th:posterior_known}, $N-k \mid \bm{Y} \sim \mathrm{Poisson}\{\lambda\, p_0(\bm{\alpha},\bm{\theta},\bm{n})\}$, so that marginalizing the first term through Lemma \ref{lemma:sm_thinning} gives
\begin{align*}
     C_{s_q,s_r}^{(n_q,n_r)} \mid \bm{Y} \overset{d}{=}\;
     &\mathrm{Poisson}\{\lambda\, p_0(\bm{\alpha},\bm{\theta},\bm{n})\, p_{1,q}(\alpha_q,\theta_q+n_q,s_q)\, p_{1,r}(\alpha_r,\theta_r+n_r,s_r)\}\\
     +\;&\mathrm{Binomial}\{k-k_{q,r},\, p_{1,q}(\alpha_q,\theta_q+n_q,s_q)\, p_{1,r}(\alpha_r,\theta_r+n_r,s_r)\}\\
     +\;&\mathrm{Binomial}\{k_q-c_{q,r},\, p_{1,r}(\alpha_r,\theta_r+n_r,s_r)\}\\
     +\;&\mathrm{Binomial}\{k_r-c_{q,r},\, p_{1,q}(\alpha_q,\theta_q+n_q,s_q)\}.
\end{align*}

\section{Derivation of the pEFPF} \label{sec:supp_marginal}

In this section we derive the partially exchangeable feature probability function in \eqref{eq:pEFPF} of the main paper, namely the probability of observing a given binary incidence array.
Let $\tilde M_{jq} = \sum_{i=1}^{n_q} \tilde Y_{ijq}$, for $j = 1, \dots, N$ and $q = 1, \dots, Q$, be the number of samples from site $q$ in which species $\tilde X_j$ is recorded, and let $\tilde{\boldsymbol m} = (\tilde m_{jq})$ be the frequencies computed from a realization $\tilde{\boldsymbol y}$ of $\tilde{\boldsymbol Y}$. We adopt the convention that the observed array $\boldsymbol Y$ is obtained from $\tilde{\boldsymbol Y}$ by discarding the species that are never recorded while preserving the original order of the remaining ones, so that the observed frequencies $\boldsymbol m$ coincide with the nonzero rows of $\tilde{\boldsymbol m}$, in the same order.

Conditionally on $N$ and on $\boldsymbol\pi = (\pi_{jq} : j = 1, \dots, N;\
q = 1, \dots, Q)$, the entries of $\tilde{\boldsymbol Y}$ are independent Bernoulli
random variables, so that the probability of a given binary array
$\tilde{\boldsymbol y}$ is
\begin{equation*}
\mathrm{pr}(\tilde{\boldsymbol Y} = \tilde{\boldsymbol y} \mid \boldsymbol\pi, N)
= \prod_{q=1}^{Q} \prod_{j=1}^{N} \pi_{jq}^{\tilde m_{jq}}
(1 - \pi_{jq})^{n_q - \tilde m_{jq}},
\end{equation*}
which depends on $\tilde{\boldsymbol y}$ only through $\tilde{\boldsymbol m}$. This is the probability of a specific binary array, not of the event $\tilde{\boldsymbol M} = \tilde{\boldsymbol m}$, which is the union of $\prod_{q=1}^{Q} \prod_{j=1}^{N} \binom{n_q}{\tilde m_{jq}}$ distinct arrays sharing the same frequencies, each with the probability above. Integrating over $\boldsymbol\pi$ we obtain
\begin{align*}
\mathrm{pr}(\tilde{\boldsymbol Y} = \tilde{\boldsymbol y} \mid N, \boldsymbol\alpha,
\boldsymbol\theta)
&= \prod_{q=1}^{Q} \prod_{j=1}^{N} \int_0^1 \pi_{jq}^{\tilde m_{jq}}
(1-\pi_{jq})^{n_q - \tilde m_{jq}}
\frac{\Gamma(\theta_q)}{\Gamma(\theta_q+\alpha_q)\Gamma(-\alpha_q)}
\pi_{jq}^{-\alpha_q - 1} (1-\pi_{jq})^{\alpha_q + \theta_q - 1}
\,\mathrm{d}\pi_{jq} \\
&= \prod_{q=1}^Q \bigg[ \bigg\{ \frac{\Gamma(\theta_q)}{\Gamma(\theta_q+ \alpha_q)\Gamma(-\alpha_q)\Gamma(\theta_q+ n_q)} \bigg\}^N \prod_{j=1}^N \Gamma(\tilde{m}_{jq}- \alpha_q) \Gamma(n_q -\tilde{m}_{jq}+ \alpha_q+\theta_q)\bigg].
\end{align*}
The event $(\boldsymbol Y = \boldsymbol y, K_{\boldsymbol n} = k)$ occurs if and only if exactly $k$ of the $N$ species are recorded and, taken in their original order, form the array $\boldsymbol y$, the remaining $N - k$ species being never recorded. There are $\binom{N}{k}$ choices for the positions of the recorded species in the complete list $\tilde X_1, \dots, \tilde X_N$, and they all have the same probability because the occurrence probabilities are independent and identically distributed across species. It therefore suffices to evaluate the display above at an array whose first $k$ species coincide with $\boldsymbol y$ and whose remaining $N - k$ species are never recorded, that is, at $(\tilde m_{1q}, \dots, \tilde m_{Nq}) = (m_{1q}, \dots, m_{kq}, 0, \dots, 0)$ for every $q = 1, \dots, Q$, and to multiply by $\binom{N}{k}$. Expressing the gamma ratios through Pochhammer symbols, we obtain, for $N \geq k$,
\begin{align*}
\mathrm{pr}(\boldsymbol Y = \boldsymbol y, K_{\boldsymbol n} = k \mid N,
\boldsymbol\alpha, \boldsymbol\theta)
&= \binom{N}{k} \prod_{q=1}^{Q} \bigg[ \left\{ \frac{\Gamma(\theta_q)}
{\Gamma(\theta_q+\alpha_q)\Gamma(-\alpha_q)\Gamma(\theta_q+n_q)} \right\}^{N}\\
&\qquad \times \{\Gamma(-\alpha_q)\Gamma(n_q+\alpha_q+\theta_q)\}^{N-k} \\
&\qquad \times \prod_{j=1}^{k} \Gamma(m_{jq}-\alpha_q)\,
\Gamma(n_q - m_{jq} + \alpha_q + \theta_q) \bigg] \\
&= \binom{N}{k} \, p_0(\boldsymbol\alpha, \boldsymbol\theta, \boldsymbol n)^{N-k}\\
&\qquad \times \prod_{q=1}^{Q} \bigg[ \left\{ \frac{-\alpha_q}{(\theta_q)_{n_q}} \right\}^{k}
\prod_{j=1}^{k} (1-\alpha_q)_{m_{jq}-1} \, (\alpha_q+\theta_q)_{n_q - m_{jq}}
\bigg],
\end{align*}
and zero otherwise, where $p_0(\boldsymbol\alpha, \boldsymbol\theta, \boldsymbol n) = \prod_{q=1}^{Q} (\alpha_q+\theta_q)_{n_q}/(\theta_q)_{n_q}$ as in the main paper.
Finally, marginalizing with respect to $N \sim \mathrm{Poisson}(\lambda)$ gives
\begin{align*}
\mathrm{pr}(\boldsymbol Y = \boldsymbol y, K_{\boldsymbol n} = k \mid \lambda,
\boldsymbol\alpha, \boldsymbol\theta)
&= \sum_{N \geq k} \frac{e^{-\lambda} \lambda^{N}}{N!} \binom{N}{k}
p_0(\boldsymbol\alpha, \boldsymbol\theta, \boldsymbol n)^{N-k} \\
&\qquad \times \prod_{q=1}^{Q} \bigg[ \left\{ \frac{-\alpha_q}{(\theta_q)_{n_q}}
\right\}^{k} \prod_{j=1}^{k} (1-\alpha_q)_{m_{jq}-1} \,
(\alpha_q+\theta_q)_{n_q - m_{jq}} \bigg] \\
&= \frac{e^{-\lambda} \lambda^{k}}{k!} \bigg[ \sum_{N \geq k}
\frac{\{\lambda \, p_0(\boldsymbol\alpha, \boldsymbol\theta, \boldsymbol n)\}^{N-k}}
{(N-k)!} \bigg] \\
&\qquad \times \prod_{q=1}^{Q} \bigg[ \left\{ \frac{-\alpha_q}{(\theta_q)_{n_q}}
\right\}^{k} \prod_{j=1}^{k} (1-\alpha_q)_{m_{jq}-1} \,
(\alpha_q+\theta_q)_{n_q - m_{jq}} \bigg] \\
&= \frac{\lambda^{k}}{k!} \exp\bigg[ -\lambda \bigg\{ 1 - \prod_{q=1}^{Q}
\frac{(\alpha_q+\theta_q)_{n_q}}{(\theta_q)_{n_q}} \bigg\} \bigg] \\
&\qquad \times \prod_{q=1}^{Q} \bigg[ \left\{ \frac{-\alpha_q}{(\theta_q)_{n_q}}
\right\}^{k} \prod_{j=1}^{k} (1-\alpha_q)_{m_{jq}-1} \,
(\alpha_q+\theta_q)_{n_q - m_{jq}} \bigg].
\end{align*}

\section{Species richness and shared species for latent occupancies} \label{sec:supp_TheLatent}

In this section, we derive results analogous to those in Section~\ref{sec:theory} of the main paper, but based on
the latent occupancies $\bm{W}$. Let $\mathcal{W}_q$ be the set of distinct species present in at least one of the $n_q$ samples from site $q$ according to $\tilde{\bm{W}}$, and let $L_{q,n_q}=|\mathcal{W}_q|=\ell_q$ be its cardinality. We write $L_{\bm{n}}=|\bigcup_{q=1}^Q \mathcal{W}_q|=\ell$ for the number of species present in at least one sample across sites, and $U_{n_q,n_r}=|\mathcal{W}_q \cap \mathcal{W}_r|=u_{q,r}$, with $q \neq r$, for the number of species present at both sites $q$ and $r$. We then let $\bm{W}=(W_{ijq})$, with $i=1,\dots,n_q$, $j=1,\dots,\ell$ and $q=1,\dots,Q$, be the array obtained from $\tilde{\bm{W}}$ by discarding the species that are never present, so that each retained species has at least one non-zero entry. Unlike their observed counterparts, none of these quantities is available in practice, and they satisfy $K_{q,n_q} \le L_{q,n_q}$, $K_{\bm{n}} \le L_{\bm{n}} \le N$ and $C_{n_q,n_r} \le U_{n_q,n_r}$.

In the terminology above, $K_{q,n_q}$ and $L_{q,n_q}$ are the observed and latent $\alpha$-diversities of site $q$, whereas $K_{\bm{n}}$ and $L_{\bm{n}}$ are their $\gamma$-diversity analogues; all four depend on the sampling effort and are bounded by the size $N$ of the species pool.
Recall also that $p_{0,q}(\alpha_q, \theta_q, n_q) = (\alpha_q + \theta_q)_{n_q}/(\theta_q)_{n_q}$ is the probability that a species goes unrecorded in the $n_q$ samples from site $q$, and that $p_{0}(\bm{\alpha}, \bm{\theta}, \bm{n}) = \prod_{q=1}^Q p_{0,q}(\alpha_q, \theta_q, n_q)$ is the probability that it goes unrecorded at every site. Conversely, we let $p_{1,q}(\alpha_q, \theta_q, n_q) := 1 - p_{0,q}(\alpha_q, \theta_q, n_q)$ be the probability that a species is recorded at least once at site $q$, and $p_{1}(\bm{\alpha}, \bm{\theta}, \bm{n}) := 1 - p_{0}(\bm{\alpha}, \bm{\theta}, \bm{n})$ the probability that it is recorded at least once overall. Finally, we write $\theta_q(\sigma_q) := \sigma_q(\alpha_q + \theta_q) - \alpha_q$ and $\bm{\theta}(\bm{\sigma}) := \{\theta_1(\sigma_1), \dots, \theta_Q(\sigma_Q)\}$, noting that $\theta_q(1) = \theta_q$. The following proposition provides the marginal distributions of these quantities.

\begin{proposition}[\emph{A priori} distinct species] \label{pr:Kprior_lat}
Under the \texttt{MOSAIC} model with parameters $(\lambda, \bm{\alpha}, \bm{\theta}, \bm{\sigma})$, the number of distinct species present at site $q$ and across sites is distributed as
\begin{equation*}
L_{q,n_q} \sim \mathrm{Poisson}\{\lambda \, p_{1,q}(\alpha_q, \theta_q(\sigma_q), n_q) \}, \quad L_{\bm{n}} \sim \mathrm{Poisson}\{\lambda \, p_1(\bm{\alpha}, \bm{\theta}(\bm{\sigma}), \bm{n}) \}, \quad (q=1,\dots,Q).
\end{equation*}
\end{proposition}

\begin{proof}
For
$$L_{q,n_q} = \sum_{j=1}^N \mathds{1}( \sum_{i=1}^{n_q} \tilde{W}_{ijq}>0) \quad \text{and} \quad L_{\bm{n}} = \sum_{j=1}^N \mathds{1}(\sum_{q=1}^Q \sum_{i=1}^{n_q} \tilde{W}_{ijq}>0),$$
the proof proceeds in exactly the same way as the one for Proposition~\ref{pr:Kprior}, with $\pi_{jq}$ replaced by the occupancy probabilities $\eta_{jq} \overset{\textup{ind}}{\sim} \mathrm{Beta}\{-\alpha_q, \sigma_q(\alpha_q+\theta_q)\}$, which are independent across $j$ and across sites. Since $\sigma_q(\alpha_q+\theta_q) = \alpha_q + \theta_q(\sigma_q)$, we have $E\{(1-\eta_{1q})^{n_q}\} = p_{0,q}(\alpha_q,\theta_q(\sigma_q),n_q)$, and it follows that
$L_{q,n_q} \sim \mathrm{Poisson}\{\lambda\, p_{1,q}(\alpha_q,\theta_q(\sigma_q),n_q)\}$ and
$L_{\bm{n}} \sim \mathrm{Poisson}\{\lambda\, p_1(\bm{\alpha},\bm{\theta}(\bm{\sigma}),\bm{n})\}$.
\end{proof}

Proposition~\ref{pr:Kprior_lat} shows that the laws of $K_{q,n_q}$ and $K_{\bm{n}}$ are recovered from those of $L_{q,n_q}$ and $L_{\bm{n}}$ in the limit $\sigma_q \to 1$, that is, under perfect detection.

Let $C_{n_q,n_r}$ denote the number of species shared between two sites $q$ and $r$, together with its latent counterpart $U_{n_q,n_r}$ based on the occupancies $\bm{W}$, so that $C_{n_q,n_r} \le U_{n_q,n_r}$. The quantity $U_{n_q,n_r}$ counts the species genuinely present at both sites.

\begin{proposition}[\emph{A priori} shared species] \label{pr:shared_lat}
Under the \texttt{MOSAIC} model with parameters $(\lambda, \bm{\alpha}, \bm{\theta}, \bm{\sigma})$, the latent number of species shared between sites $q$ and $r$, with $q \neq r$, is distributed as
\begin{equation*}
U_{n_q,n_r} \sim \mathrm{Poisson}\{\lambda \, p_{1,q}(\alpha_q, \theta_q(\sigma_q), n_q) \, p_{1,r}(\alpha_r, \theta_r(\sigma_r), n_r)\}.
\end{equation*}
\end{proposition}

\begin{proof}
For the latent number of shared species
\begin{equation*}
    U_{n_q,n_r} = \sum_{j=1}^N \mathds{1}(\sum_{i=1}^{n_q} \tilde{W}_{ijq}>0) \mathds{1}(\sum_{i=1}^{n_r} \tilde{W}_{ijr}>0),
\end{equation*}
the proof proceeds in the same way as the one for Proposition~\ref{pr:shared} with $\pi_{jq}, \pi_{jr}$ replaced by the occupancy probabilities $\eta_{jq} \overset{\textup{ind}}{\sim} \mathrm{Beta}\{-\alpha_q, \sigma_q(\alpha_q+\theta_q)\}$ and $\eta_{jr} \overset{\textup{ind}}{\sim} \mathrm{Beta}\{-\alpha_r, \sigma_r(\alpha_r+\theta_r)\}$, which are independent across $j$ and across sites. Since $\sigma_q(\alpha_q+\theta_q) = \alpha_q + \theta_q(\sigma_q)$, we have $E\{(1-\eta_{1q})^{n_q}\} = p_{0,q}(\alpha_q,\theta_q(\sigma_q),n_q)$, giving
$$
U_{n_q,n_r} \sim \mathrm{Poisson}\{\lambda\, p_{1,q}(\alpha_q,\theta_q(\sigma_q),n_q)\,
p_{1,r}(\alpha_r,\theta_r(\sigma_r),n_r)\}.
$$
\end{proof}

The same reading of the previous results applies to $U_{n_q,n_r}$, that coincides to $C_{n_q,n_r}$ in the limit $\sigma_q, \sigma_r \to 1$.

Propositions~\ref{pr:Kprior_lat}-\ref{pr:shared_lat} concern \emph{a priori} properties of $L_{q,n_q}$, $L_{\bm{n}}$ and $U_{n_q,n_r}$. We now tackle
the problem of predicting the number of new species that would be discovered if $\bm{s} = (s_1,\dots,s_Q)$ additional samples were collected after an initial $\bm{n}$, resulting in an enlarged sample of size $\bm{n} + \bm{s}$.
We denote by $K_{q,s_q}^{(n_q)}$ and $K_{\bm{s}}^{(\bm{n})}$ the number of new species that would be discovered at site $q$ and across all sites, respectively, and by $L_{q,s_q}^{(n_q)}$ and $L_{\bm{s}}^{(\bm{n})}$ their latent counterparts based on $\bm{W}$. These quantities extrapolate the accumulation curves beyond the observed effort, since $E(K_{\bm{n}+\bm{s}} \mid \bm{Y}) = k + E(K_{\bm{s}}^{(\bm{n})} \mid \bm{Y})$. Two groups of species contribute to the predictive of $K_{q,s_q}^{(n_q)}$: those not yet recorded anywhere, whose number is random; and the $k - k_q$ species recorded at some other site but not at $q$, whose number is known given $\bm{Y}$. The same reading applies to $L_{q,s_q}^{(n_q)}$, with $\ell$ and $\ell_q$ in place of $k$ and $k_q$.

\begin{theorem}[\emph{A posteriori} distinct species] \label{thm:Kpred_lat}
Under the \texttt{MOSAIC} model with parameters $(\lambda, \bm{\alpha}, \bm{\theta}, \bm{\sigma})$, the latent number of new species discovered at site $q$ and across sites satisfies
\begin{equation*}
\begin{aligned}
L_{q,s_q}^{(n_q)} \mid \bm{W} \overset{d}{=} \; & \mathrm{Poisson}\{\lambda \, p_{0}(\bm{\alpha}, \bm{\theta}(\bm{\sigma}), \bm{n}) \, p_{1,q}(\alpha_q, \theta_q(\sigma_q) + n_q, s_q)\} \\
 \; &+ \mathrm{Binomial}\{\ell - \ell_q, \, p_{1,q}(\alpha_q, \theta_q(\sigma_q) + n_q, s_q)\}, \quad (q = 1,\dots,Q),\\
L_{\bm{s}}^{(\bm{n})} \mid \bm{W} \; \sim \; & \mathrm{Poisson}\{\lambda \, p_{0}(\bm{\alpha}, \bm{\theta}(\bm{\sigma}), \bm{n}) \, p_{1}(\bm{\alpha}, \bm{\theta}(\bm{\sigma}) + \bm{n}, \bm{s})\}.
\end{aligned}
\end{equation*}
\end{theorem}

\begin{proof}
The latent quantities are obtained along the same lines of the proof of Theorem~\ref{thm:Kpred}, replacing the observed data $\bm{Y}$ with the occupancies $\bm{W}$ throughout. Let $\mathcal{B} = \{ j \in \{1,\dots,N\} : \sum_{q=1}^Q \sum_{i=1}^{n_q} \tilde{W}_{ijq}=0 \}$ and $\mathcal{B}_q = \{ j \in \{1,\dots,N\} : \sum_{i=1}^{n_q} \tilde{W}_{ijq}=0 \}$ be the sets of species absent from all samples and from those collected at site $q$, respectively, so that $|\mathcal{B}| = N-\ell$ and $|\mathcal{B}_q| = N-\ell_q$, and set
\begin{equation*}
    L_{\bm{s}}^{(\bm{n})} = \sum_{j \in \mathcal{B}} \mathds{1}\Big(\sum_{q=1}^Q\sum_{i=n_q+1}^{n_q+s_q} \tilde{W}_{ijq}>0\Big),
    \qquad
    L_{q,s_q}^{(n_q)} = \sum_{j \in \mathcal{B}_q} \mathds{1}\Big(\sum_{i=n_q+1}^{n_q+s_q} \tilde{W}_{ijq}>0\Big).
\end{equation*}
Conditioning on $\bm{W}$ rather than on $\bm{Y}$ is essential here: a null entry of $\bm{Y}$ is ambiguous between absence and non-detection, whereas a null entry of $\bm{W}$ identifies absence, and only in the latter case does the occupancy probability of an unrecorded species admit the conjugate posterior $\eta_{jq} \mid \bm{W} \overset{\textup{ind}}{\sim} \mathrm{Beta}\{-\alpha_q, \sigma_q(\alpha_q+\theta_q)+n_q\}$. Since $\sigma_q(\alpha_q+\theta_q) = \alpha_q + \theta_q(\sigma_q)$, the argument used for the occurrence probabilities applies verbatim with $\theta_q$ replaced by $\theta_q(\sigma_q)$, and the number of species absent everywhere satisfies $N-\ell \mid \bm{W} \sim \mathrm{Poisson}\{\lambda\, p_0(\bm{\alpha},\bm{\theta}(\bm{\sigma}),\bm{n})\}$. Hence
\begin{equation*}
    L_{\bm{s}}^{(\bm{n})} \mid \bm{W} \sim \mathrm{Poisson}\{\lambda\, p_0(\bm{\alpha},\bm{\theta}(\bm{\sigma}),\bm{n})\, p_1(\bm{\alpha},\bm{\theta}(\bm{\sigma})+\bm{n},\bm{s})\},
\end{equation*}
and, decomposing $\mathcal{B}_q = \mathcal{B} \cup (\mathcal{B}_q \setminus \mathcal{B})$ with $|\mathcal{B}_q \setminus \mathcal{B}| = \ell-\ell_q$ known given $\bm{W}$,
\begin{align*}
    L_{q,s_q}^{(n_q)} \mid \bm{W} \overset{d}{=}\;
    &\mathrm{Poisson}\{\lambda\, p_0(\bm{\alpha},\bm{\theta}(\bm{\sigma}),\bm{n})\, p_{1,q}(\alpha_q,\theta_q(\sigma_q)+n_q,s_q)\}\\
    +\;&\mathrm{Binomial}\{\ell-\ell_q,\, p_{1,q}(\alpha_q,\theta_q(\sigma_q)+n_q,s_q)\},
\end{align*}
with the two terms independent, as above.
\end{proof}

These quantities are conditional on $\bm{W}$, which is itself unobserved, and are therefore available only within the sampler of Section \ref{sec:computation}, where $\tilde{\bm{W}}$ is imputed at each iteration.

Finally, we characterize the predictive distribution of the number of new latent species shared between two sites $q$ and $r$, denoted by $U_{s_q,s_r}^{(n_q,n_r)}$ based on~$\bm{W}$. Let $k_{q,r} = k_q + k_r - c_{q,r}$ be the number of distinct species recorded at site $q$ or at site $r$, and let $\ell_{q,r} = \ell_q + \ell_r - u_{q,r}$ be its latent counterpart.

\begin{theorem}[\emph{A posteriori} shared species] \label{thm:pred_shared_lat}
Under the \texttt{MOSAIC} model with parameters $(\lambda, \bm{\alpha}, \bm{\theta}, \bm{\sigma})$, the latent number of species shared between sites $q$ and $r$, with $q \neq r$, that would be discovered in $(s_q,s_r)$ additional samples satisfies
\begin{equation*}
\begin{aligned}
U_{s_q,s_r}^{(n_q,n_r)} \mid \bm{W} \overset{d}{=} \; & \mathrm{Poisson}\{\lambda \, p_0(\bm{\alpha}, \bm{\theta}(\bm{\sigma}), \bm{n}) \, p_{1,q}(\alpha_q, \theta_q(\sigma_q) + n_q, s_q) \, p_{1,r}(\alpha_r, \theta_r(\sigma_r) + n_r, s_r)\} \\
+ \; & \mathrm{Binomial}\{\ell - \ell_{q,r}, \, p_{1,q}(\alpha_q, \theta_q(\sigma_q) + n_q, s_q) \, p_{1,r}(\alpha_r, \theta_r(\sigma_r) + n_r, s_r)\} \\
+ \; & \mathrm{Binomial}\{\ell_q - u_{q,r}, \, p_{1,r}(\alpha_r, \theta_r(\sigma_r) + n_r, s_r)\} \\
+ \; & \mathrm{Binomial}\{\ell_r - u_{q,r}, \, p_{1,q}(\alpha_q, \theta_q(\sigma_q) + n_q, s_q)\},
\end{aligned}
\end{equation*}
where the four terms are independent.
\end{theorem}

\begin{proof}
The proof proceeds along the same lines of the proof of Theorem~\ref{thm:pred_shared}, replacing the observed data $\bm{Y}$ with the occupancies $\bm{W}$ throughout. The three sets are now defined by absence rather than by non-detection, namely
\begin{equation*}
\begin{aligned}
    \mathcal{B}_{00} &= \Big\{j \in \{1,\dots,N\} : \sum_{i=1}^{n_q}\tilde{W}_{ijq}=0,\; \sum_{i=1}^{n_r}\tilde{W}_{ijr}=0\Big\},\\
    \mathcal{B}_{q} &= \Big\{j \in \{1,\dots,N\} : \sum_{i=1}^{n_q}\tilde{W}_{ijq}>0,\; \sum_{i=1}^{n_r}\tilde{W}_{ijr}=0\Big\},\\
    \mathcal{B}_{r} &= \Big\{j \in \{1,\dots,N\} : \sum_{i=1}^{n_q}\tilde{W}_{ijq}=0,\; \sum_{i=1}^{n_r}\tilde{W}_{ijr}>0\Big\};
\end{aligned}
\end{equation*}
the target is \eqref{eq:sm_target_cpqpred} with the indicators $\tilde{I}_{jh} = \mathds{1}(\sum_{i=n_h+1}^{n_h+s_h} \tilde{W}_{ijh}>0)$ in place of the $I_{jh}$; and the occupancy probabilities of the species absent from site $h$ satisfy $\eta_{jh} \mid \bm{W} \overset{\textup{ind}}{\sim} \mathrm{Beta}\{-\alpha_h, \sigma_h(\alpha_h+\theta_h)+n_h\}$. Since $\sigma_h(\alpha_h+\theta_h) = \alpha_h + \theta_h(\sigma_h)$, every occurrence of $\theta_h$ is replaced by $\theta_h(\sigma_h)$, including in the Poisson mean, which becomes $\lambda\, p_0(\bm{\alpha},\bm{\theta}(\bm{\sigma}),\bm{n})$ because $N-\ell \mid \bm{W} \sim \mathrm{Poisson}\{\lambda\, p_0(\bm{\alpha},\bm{\theta}(\bm{\sigma}),\bm{n})\}$. The four cardinalities become $N-\ell$, $\ell-\ell_{q,r}$, $\ell_q-u_{q,r}$ and $\ell_r-u_{q,r}$, with $\ell_{q,r} = \ell_q+\ell_r-u_{q,r}$, all determined by $\bm{W}$.
\end{proof}

\section{The negative binomial extension} \label{sec:supp_nb}

\subsection{Extension of theoretical results}

In this section we adapt the results of Section~\ref{sec:theory} of the main paper to the case in which the total number of species follows a negative binomial distribution, i.e. $N \sim \mathrm{NegBinomial}(\rho_0,\nu_0)$ with size $\rho_0>0$ and mean $\nu_0>0$. This is equivalent to assuming $N\sim \mathrm{Poisson}(\lambda)$ and $\lambda \sim \mathrm{Gamma}(\rho_0, \rho_0/ \nu_0)$. All the remaining assumptions of the \texttt{MOSAIC} model are left unchanged.
The results below are obtained exactly as their Poisson counterparts, the only difference being that
Lemma \ref{lemma:sm_thinning} is replaced by the following thinning property of the negative
binomial distribution.

\begin{lemma} \label{lemma:sm_nbthinning}
If $V \sim \mathrm{NegBinomial}(r,\mu)$ and $X \mid V \sim \mathrm{Binomial}(V,p)$, then $X \sim \mathrm{NegBinomial}(r,\mu p)$.
\end{lemma}

\begin{proof}
Let $\phi = \mu/(\mu+r)$, so that the probability generating function of $V$ is $E(z^V) = \{(1-\phi)/(1-\phi z)\}^{r}$. Then
\begin{equation*}
    E(z^X) = E\{E(z^X \mid V)\} = E\{(1-p+pz)^V\}
    = \Big(\frac{1-\phi}{1-\phi(1-p)-\phi p z}\Big)^{r}
    = \Big(\frac{1-\phi'}{1-\phi' z}\Big)^{r},
\end{equation*}
with $\phi' = \phi p /\{1-\phi(1-p)\}$, since $1-\phi' = (1-\phi)/\{1-\phi(1-p)\}$. Hence $X$ is negative binomial with size $r$ and mean $r\phi'/(1-\phi') = r\phi p/(1-\phi) = \mu p$.
\end{proof}

We first state the analogue of Proposition~\ref{th:posterior_known}. Recall that $N' = N-k$ denotes the number of species unseen given $\bm{Y}$ and that $L' = N - \ell$ denotes the number of species never present given $\bm{W}$.

\begin{proposition} \label{prop:sm_nbpost}
Let $\bm{Y}$ follow the negative binomial specification of the \texttt{MOSAIC} model with parameters $(\rho_0, \nu_0, \bm{\alpha},\bm{\theta},\bm{\sigma})$, then
\begin{equation*}
    N' \mid \bm{Y} \sim \mathrm{NegBinomial}(\rho_0+k,\, \mu_0),
    \qquad
    \mu_0 = \frac{(\rho_0+k)\, p_0(\bm{\alpha},\bm{\theta},\bm{n})}
                 {\rho_0/\nu_0 + p_1(\bm{\alpha},\bm{\theta},\bm{n})},
\end{equation*}
whereas the posterior distributions of the occurrence probabilities $\pi^{*}_{jq}$ and $\pi'_{jq}$ are those given in Proposition~\ref{th:posterior_known}, and are therefore unchanged.
\end{proposition}

\begin{proposition}[\emph{A priori} distinct species] \label{prop:sm_nbprior1}
Under the \texttt{MOSAIC} model with $N \sim \mathrm{NegBinomial}(\rho_0,\nu_0)$, the number of distinct species observed at site $q$ and across sites is distributed as
\begin{equation*}
    K_{q,n_q} \sim \mathrm{NegBinomial}\{\rho_0,\, \nu_0\, p_{1,q}(\alpha_q,\theta_q,n_q)\},
    \qquad
    K_{\bm{n}} \sim \mathrm{NegBinomial}\{\rho_0,\, \nu_0\, p_1(\bm{\alpha},\bm{\theta},\bm{n})\},
\end{equation*}
for $q = 1,\dots,Q$, whereas the number of distinct species present at site $q$ and across sites is distributed as
\begin{equation*}
    L_{q,n_q} \sim \mathrm{NegBinomial}\{\rho_0,\, \nu_0\, p_{1,q}(\alpha_q,\theta_q(\sigma_q),n_q)\},
    \qquad
    L_{\bm{n}} \sim \mathrm{NegBinomial}\{\rho_0,\, \nu_0\, p_1(\bm{\alpha},\bm{\theta}(\bm{\sigma}),\bm{n})\}.
\end{equation*}
\end{proposition}

\begin{proof}
The proof of Proposition~\ref{pr:Kprior} shows that $K_{q,n_q} \mid N \sim \mathrm{Binomial}\{N, p_{1,q}(\alpha_q,\theta_q,n_q)\}$ and
$K_{\bm{n}} \mid N \sim \mathrm{Binomial}\{N, p_1(\bm{\alpha},\bm{\theta},\bm{n})\}$, since that step does not involve the law of $N$. The statements then follow from Lemma \ref{lemma:sm_nbthinning} in place of Lemma \ref{lemma:sm_thinning}. For the latent counterparts the same proof applies with $\theta_q$ replaced by $\theta_q(\sigma_q)$.
\end{proof}

The simple interpretation of the Poisson case carries over: $E(K_{\bm{n}}) = \nu_0\, p_1(\bm{\alpha},\bm{\theta},\bm{n})$, with $\nu_0$ the expected total number of species and $p_1(\bm{\alpha},\bm{\theta},\bm{n})$ the expected fraction of these that is observed under sampling effort $\bm{n}$.

\begin{proposition}[\emph{A priori} shared species] \label{prop:sm_nbprior2}
Under the \texttt{MOSAIC} model with $N \sim \mathrm{NegBinomial}(\rho_0,\nu_0)$, the number of species shared between sites $q$ and $r$, with $q \neq r$, is distributed as
\begin{equation*}
    C_{n_q,n_r} \sim \mathrm{NegBinomial}\{\rho_0,\, \nu_0\, p_{1,q}(\alpha_q,\theta_q,n_q)\, p_{1,r}(\alpha_r,\theta_r,n_r)\},
\end{equation*}
whereas its latent counterpart is distributed as
\begin{equation*}
    U_{n_q,n_r} \sim \mathrm{NegBinomial}\{\rho_0,\, \nu_0\, p_{1,q}(\alpha_q,\theta_q(\sigma_q),n_q)\, p_{1,r}(\alpha_r,\theta_r(\sigma_r),n_r)\}.
\end{equation*}
\end{proposition}

\begin{proof}
As in the proof of Proposition~\ref{pr:shared}, $C_{n_q,n_r} \mid N \sim \mathrm{Binomial}\{N, p_{1,q}(\alpha_q,\theta_q,n_q) p_{1,r}(\alpha_r,\theta_r,n_r)\}$, and the statement follows from Lemma \ref{lemma:sm_nbthinning}; the latent counterpart is obtained by replacing $\theta_q$ and $\theta_r$ with $\theta_q(\sigma_q)$ and $\theta_r(\sigma_r)$ respectively.
\end{proof}

\begin{proposition}[\emph{A posteriori} distinct species] \label{prop:sm_nbpost1}
Under the \texttt{MOSAIC} model with $N \sim \mathrm{NegBinomial}(\rho_0,\nu_0)$, the number of new species discovered at site $q$ and across sites satisfies
\begin{align*}
    K_{q,s_q}^{(n_q)} \mid \bm{Y} \overset{d}{=}\;
    &\mathrm{NegBinomial}\{\rho_0+k,\, \mu_0\, p_{1,q}(\alpha_q,\theta_q+n_q,s_q)\}\\
    +\;&\mathrm{Binomial}\{k-k_q,\, p_{1,q}(\alpha_q,\theta_q+n_q,s_q)\}, \qquad (q = 1,\dots,Q),\\
    K_{\bm{s}}^{(\bm{n})} \mid \bm{Y} \sim\;
    &\mathrm{NegBinomial}\{\rho_0+k,\, \mu_0\, p_1(\bm{\alpha},\bm{\theta}+\bm{n},\bm{s})\},
\end{align*}
with $\mu_0$ as in Proposition \ref{prop:sm_nbpost}, whereas their latent counterparts satisfy
\begin{align*}
    L_{q,s_q}^{(n_q)} \mid \bm{W} \overset{d}{=}\;
    &\mathrm{NegBinomial}\{\rho_0+\ell,\, \mu_0^{W}\, p_{1,q}(\alpha_q,\theta_q(\sigma_q)+n_q,s_q)\}\\
    +\;&\mathrm{Binomial}\{\ell-\ell_q,\, p_{1,q}(\alpha_q,\theta_q(\sigma_q)+n_q,s_q)\}, \qquad (q = 1,\dots,Q),\\
    L_{\bm{s}}^{(\bm{n})} \mid \bm{W} \sim\;
    &\mathrm{NegBinomial}\{\rho_0+\ell,\, \mu_0^{W}\, p_1(\bm{\alpha},\bm{\theta}(\bm{\sigma})+\bm{n},\bm{s})\},
\end{align*}
with
$$\mu_0^{W} = \frac{(\rho_0+\ell)\, p_0(\bm{\alpha},\bm{\theta}(\bm{\sigma}),\bm{n})}
                     {\rho_0/\nu_0 + p_1(\bm{\alpha},\bm{\theta}(\bm{\sigma}),\bm{n})}.$$
\end{proposition}

\begin{proof}
The proof of Theorem~\ref{thm:Kpred} applies up to the point where the number of unseen species is integrated out: conditionally on $N'$ one has $K_{\bm{s}}^{(\bm{n})} \mid N', \bm{Y} \sim \mathrm{Binomial}\{N', p_1(\bm{\alpha},\bm{\theta}+\bm{n},\bm{s})\}$, and the species new to site $q$ split into the $N'$ species unrecorded everywhere and the $k-k_q$
species recorded at some other site but not at $q$, each contributing with probability $p_{1,q}(\alpha_q,\theta_q+n_q,s_q)$. Since $N' \mid \bm{Y} \sim \mathrm{NegBinomial}(\rho_0+k,\mu_0)$ by Proposition \ref{prop:sm_nbpost}, applying Lemma \ref{lemma:sm_nbthinning} instead of Lemma \ref{lemma:sm_thinning} to the first group gives the stated expressions, the binomial term being unchanged because $k-k_q$ is known given $\bm{Y}$. The latent counterparts are obtained analogously, by replacing $\bm{Y}$ with $\bm{W}$, $\bm{\theta}$ with $\bm{\theta}(\bm{\sigma})$ and $(k,k_q,N',\mu_0)$ with $(\ell,\ell_q,L',\mu_0^{W})$.
\end{proof}

\begin{proposition}[\emph{A posteriori} shared species] \label{prop:sm_nbpost2}
Under the \texttt{MOSAIC} model with $N \sim \mathrm{NegBinomial}(\rho_0,\nu_0)$, the number of species shared between sites $q$ and $r$, with $q \neq r$, that would be discovered in $(s_q,s_r)$ additional samples satisfies
\begin{align*}
     C_{s_q,s_r}^{(n_q,n_r)} \mid \bm{Y} \overset{d}{=}\;
     &\mathrm{NegBinomial}\{\rho_0+k,\, \mu_0\, p_{1,q}(\alpha_q,\theta_q+n_q,s_q)\, p_{1,r}(\alpha_r,\theta_r+n_r,s_r)\}\\
     +\;&\mathrm{Binomial}\{k-k_{q,r},\, p_{1,q}(\alpha_q,\theta_q+n_q,s_q)\, p_{1,r}(\alpha_r,\theta_r+n_r,s_r)\}\\
     +\;&\mathrm{Binomial}\{k_q-c_{q,r},\, p_{1,r}(\alpha_r,\theta_r+n_r,s_r)\}\\
     +\;&\mathrm{Binomial}\{k_r-c_{q,r},\, p_{1,q}(\alpha_q,\theta_q+n_q,s_q)\},
\end{align*}
with $\mu_0$ as in Proposition \ref{prop:sm_nbpost}, whereas its latent counterpart satisfies
\begin{align*}
     U_{s_q,s_r}^{(n_q,n_r)} \mid \bm{W} \overset{d}{=}\;
     &\mathrm{NegBinomial}\{\rho_0+\ell,\, \mu_0^{W}\, p_{1,q}(\alpha_q,\theta_q(\sigma_q)+n_q,s_q)\, p_{1,r}(\alpha_r,\theta_r(\sigma_r)+n_r,s_r)\}\\
     +\;&\mathrm{Binomial}\{\ell-\ell_{q,r},\, p_{1,q}(\alpha_q,\theta_q(\sigma_q)+n_q,s_q)\, p_{1,r}(\alpha_r,\theta_r(\sigma_r)+n_r,s_r)\}\\
     +\;&\mathrm{Binomial}\{\ell_q-u_{q,r},\, p_{1,r}(\alpha_r,\theta_r(\sigma_r)+n_r,s_r)\}\\
     +\;&\mathrm{Binomial}\{\ell_r-u_{q,r},\, p_{1,q}(\alpha_q,\theta_q(\sigma_q)+n_q,s_q)\},
\end{align*}
with $\mu_0^{W}$ as in Proposition \ref{prop:sm_nbpost1}.
\end{proposition}

\begin{proof}
Identical to the proof of Theorem~\ref{thm:pred_shared}: the four groups of species and the corresponding success probabilities are unaffected by the law of $N$, and only the first term, which involves the $N'$ species unrecorded everywhere, requires integrating out $N'$. Since $N' \mid \bm{Y} \sim \mathrm{NegBinomial}(\rho_0+k,\mu_0)$ by Proposition \ref{prop:sm_nbpost}, Lemma \ref{lemma:sm_nbthinning} replaces Lemma \ref{lemma:sm_thinning} in that step. The latent counterpart follows in the same way, with $\bm{\theta}$ replaced by $\bm{\theta}(\bm{\sigma})$ and $(k, k_q, k_r, c_{q,r}, k_{q,r}, N', \mu_0)$ replaced by $(\ell, \ell_q, \ell_r, u_{q,r}, \ell_{q,r}, L', \mu_0^{W})$.
\end{proof}

\subsection{Derivation of the pEFPF}

We now derive the pEFPF under $N \sim \mathrm{NegBinomial}(\rho_0, \nu_0)$.
As in Section \ref{sec:supp_marginal}, we consider the probability of the incidence array $\boldsymbol Y = \boldsymbol y$, under the same ordering convention. Then
\begin{equation}
\label{eq:sm_marginal_negbin}
\begin{split}
&\mathrm{pr}(\boldsymbol Y = \boldsymbol y, K_{\boldsymbol n} = k \mid \rho_0, \nu_0,
\boldsymbol\alpha, \boldsymbol\theta) \\
&\quad = \frac{\Gamma(k + \rho_0)}{k! \, \Gamma(\rho_0)}
\left( \frac{\rho_0}{\rho_0 + \nu_0} \right)^{\rho_0}
\left( \frac{\nu_0}{\nu_0 + \rho_0} \right)^{k}  \bigg\{ 1 - \frac{\nu_0}{\nu_0 + \rho_0} \prod_{q=1}^{Q} \frac{(\alpha_q+\theta_q)_{n_q}}{(\theta_q)_{n_q}} \bigg\}^{-k-\rho_0} \\
&\qquad \times \prod_{q=1}^{Q} \bigg[ \left\{ \frac{-\alpha_q}{(\theta_q)_{n_q}}
\right\}^{k} \prod_{j=1}^{k} (1-\alpha_q)_{m_{jq}-1} \,
(\alpha_q+\theta_q)_{n_q - m_{jq}} \bigg].
\end{split}
\end{equation}

\begin{proof}
The conditional probability $\mathrm{pr}(\boldsymbol Y = \boldsymbol y, K_{\boldsymbol n} = k \mid N, \boldsymbol\alpha, \boldsymbol\theta)$ derived in Section \ref{sec:supp_marginal} does not involve the law of $N$. Let $p_0 = p_0(\boldsymbol\alpha, \boldsymbol\theta, \boldsymbol n)$. Marginalizing
with respect to $N$,
\begin{align*}
&\mathrm{pr}(\boldsymbol Y = \boldsymbol y, K_{\boldsymbol n} = k \mid \rho_0, \nu_0,
\boldsymbol\alpha, \boldsymbol\theta) \\
&\quad = \left( \frac{\rho_0}{\rho_0 + \nu_0} \right)^{\rho_0}
\bigg[ \sum_{h \geq k} \frac{\Gamma(h + \rho_0)}{h! \, \Gamma(\rho_0)}
\left( \frac{\nu_0}{\nu_0 + \rho_0} \right)^{h} \binom{h}{k} p_0^{\,h-k} \bigg] \\
&\qquad \times \prod_{q=1}^{Q} \bigg[ \left\{ \frac{-\alpha_q}{(\theta_q)_{n_q}}
\right\}^{k} \prod_{j=1}^{k} (1-\alpha_q)_{m_{jq}-1} \,
(\alpha_q+\theta_q)_{n_q - m_{jq}} \bigg] \\
&\quad = \left( \frac{\rho_0}{\rho_0 + \nu_0} \right)^{\rho_0}
\left( \frac{\nu_0}{\nu_0 + \rho_0} \right)^{k} \frac{1}{k! \, \Gamma(\rho_0)}
\bigg[ \sum_{t \geq 0} \frac{\Gamma(t + k + \rho_0)}{t!}
\left( \frac{\nu_0 \, p_0}{\nu_0 + \rho_0} \right)^{t} \bigg] \\
&\qquad \times \prod_{q=1}^{Q} \bigg[ \left\{ \frac{-\alpha_q}{(\theta_q)_{n_q}}
\right\}^{k} \prod_{j=1}^{k} (1-\alpha_q)_{m_{jq}-1} \,
(\alpha_q+\theta_q)_{n_q - m_{jq}} \bigg].
\end{align*}
Since $0 < \nu_0 p_0/(\nu_0 + \rho_0) < 1$, the series in square brackets equals
$\Gamma(k + \rho_0)\{1 - \nu_0 p_0/(\nu_0 + \rho_0)\}^{-k-\rho_0}$, which yields
\eqref{eq:sm_marginal_negbin}.
\end{proof}

\section{Additional details on posterior computation}\label{sec:computation}

Posterior computation for the proposed \texttt{MOSAIC} model proceeds through an MCMC algorithm. The procedure is facilitated by the explicit relationship between the partially exchangeable BB and \texttt{MOSAIC} models established in equation~\eqref{eq:marginal_beta}, and by the availability of the closed-form pEFPF in \eqref{eq:pEFPF}. Indeed, draws from the posterior distribution of $(\bm{\alpha}, \bm{\theta})$ can be obtained via any collapsed Metropolis--Hastings algorithm leveraging the marginal likelihood in \eqref{eq:pEFPF}. Specifically, we obtain posterior samples of the regression coefficients $\bm{\zeta}$ and of the precision parameters $\bm{\phi}$ using Hamiltonian Monte Carlo implemented in the probabilistic programming language Stan \citep{carpenter2017stan}, and then map these samples through \eqref{eq:rip_beta} to recover posterior draws of $(\bm{\alpha}, \bm{\theta})$. We recall that the posterior of $N$ has a closed-form expression, as stated in Proposition~\ref{th:posterior_known}, so that no MCMC step is needed to obtain posterior draws of the number of unseen species. To sample the occupancy and detectability components $\bm{\eta}_q$ and $\bm{\gamma}_q$, the sampler relies on a data augmentation strategy based on the augmented arrays $\tilde{\bm{W}}$ and $\tilde{\bm{T}}$. Conditional on $N$ and on these augmented data, the full conditional distributions of $\eta_{jq}$ and $\gamma_{jq}$ are beta by conjugacy, which enables straightforward Gibbs updates. The full sampling procedure is summarized in Algorithm~\ref{alg:gibbs}, which is stated for the Poisson formulation with $\lambda$ fixed, for the sake of simplicity. The extension to a random $\lambda$ with a Gamma prior is straightforward and requires only two modifications: the marginal likelihood in step 1 and the sampling distribution of $N'$ in step 3 are replaced by their negative binomial counterparts, reported in Section \ref{sec:supp_nb}.

\begin{algorithm}[t]
\caption{Posterior sampling algorithm for the \texttt{MOSAIC} model}\label{alg:gibbs}
\begin{algorithmic}
\vspace{5pt}
\State 1. Sample $B$ draws of $\bm{\zeta}$ and $\bm{\phi}$ from their posterior through a collapsed HMC algorithm targeting the marginal likelihood in \eqref{eq:pEFPF}.
\State 2. Set $\bm{\theta}^{(b)} = \bm{\phi}^{(b)}$ and $\bm{\alpha}^{(b)} = -\bm{\mu}^{(b)}\bm{\phi}^{(b)}$, with $\bm{\mu}^{(b)} = 1/(1+e^{-\bm{Z}\bm{\zeta}^{(b)}})$, for $b=1,\ldots,B$.
\State 3. For $b=1,\ldots,B$, sample the number of unseen species $N'^{(b)} \sim \mathrm{Poisson}\{\lambda\, p_0(\bm{\alpha}^{(b)},\bm{\theta}^{(b)},\bm{n})\}$, and set $N^{(b)} = k + N'^{(b)}$.
\State 4. For each draw $b=1,\ldots,B$ and each site $q=1,\ldots,Q$ (in parallel), suppress the index $b$ for readability. Set $\tilde{y}_{ijq}=y_{ijq}$ for the observed species $j=1,\ldots,k$, and $\tilde{y}_{ijq}=0$ for the unseen species $j=k+1,\ldots,N$:
\begin{enumerate}[label=(\roman*)]
    \item For $i=1,\ldots,n_q$ and $j=1,\ldots,N$, if $\tilde{y}_{ijq}=1$, set $\tilde{W}_{ijq}=\tilde{T}_{ijq}=1$. Otherwise, sample $(\tilde{W}_{ijq},\tilde{T}_{ijq})$ from the categorical distribution
    \begin{equation*}
    (\tilde{W}_{ijq},\tilde{T}_{ijq}) =
    \begin{cases}
    (0,0), & \text{with probability proportional to } (1-\eta_{jq})(1-\gamma_{jq}),\\
    (0,1), & \text{with probability proportional to } (1-\eta_{jq})\gamma_{jq},\\
    (1,0), & \text{with probability proportional to } \eta_{jq}(1-\gamma_{jq}).
    \end{cases}
    \end{equation*}
    \item For $j=1,\ldots,N$, sample
    \begin{equation*}
    \begin{aligned}
    \eta_{jq} \mid - \;&\sim \mathrm{Beta}\Big( -\alpha_q + \textstyle\sum_{i=1}^{n_q}\tilde{W}_{ijq},\, \sigma_q(\alpha_q+\theta_q)+n_q- \textstyle\sum_{i=1}^{n_q}\tilde{W}_{ijq} \Big),\\
    \gamma_{jq} \mid - \;&\sim \mathrm{Beta}\Big( -\alpha_q+\sigma_q(\alpha_q+\theta_q)+ \textstyle\sum_{i=1}^{n_q}\tilde{T}_{ijq},\, (1-\sigma_q)(\alpha_q+\theta_q)+n_q- \textstyle\sum_{i=1}^{n_q}\tilde{T}_{ijq} \Big).
    \end{aligned}
    \end{equation*}
\end{enumerate}
\end{algorithmic}
\end{algorithm}

\section{Additional results for simulation study}

\subsection{Additional experiments}

To assess the performance of the \texttt{MOSAIC} model in predicting $K_{q,s_q}^{(n_q)}$ as a function of the number of sites, we simulate data under the same settings as scenario A in Section~\ref{sec:sim_study} of the main paper, considering different values for $Q$, namely $Q \in \{5,8,10,15,20,25,30\}$. Figure~\ref{fig:QmulMSE} shows the predictive performance compared with the exchangeable negative binomial mixture of beta Bernoulli. As $Q$ increases, the improved predictive performance of \texttt{MOSAIC} relative to analyzing data from each site separately becomes increasingly apparent.

\begin{figure}[t]
    \centering
    \includegraphics[width=0.76\textwidth]{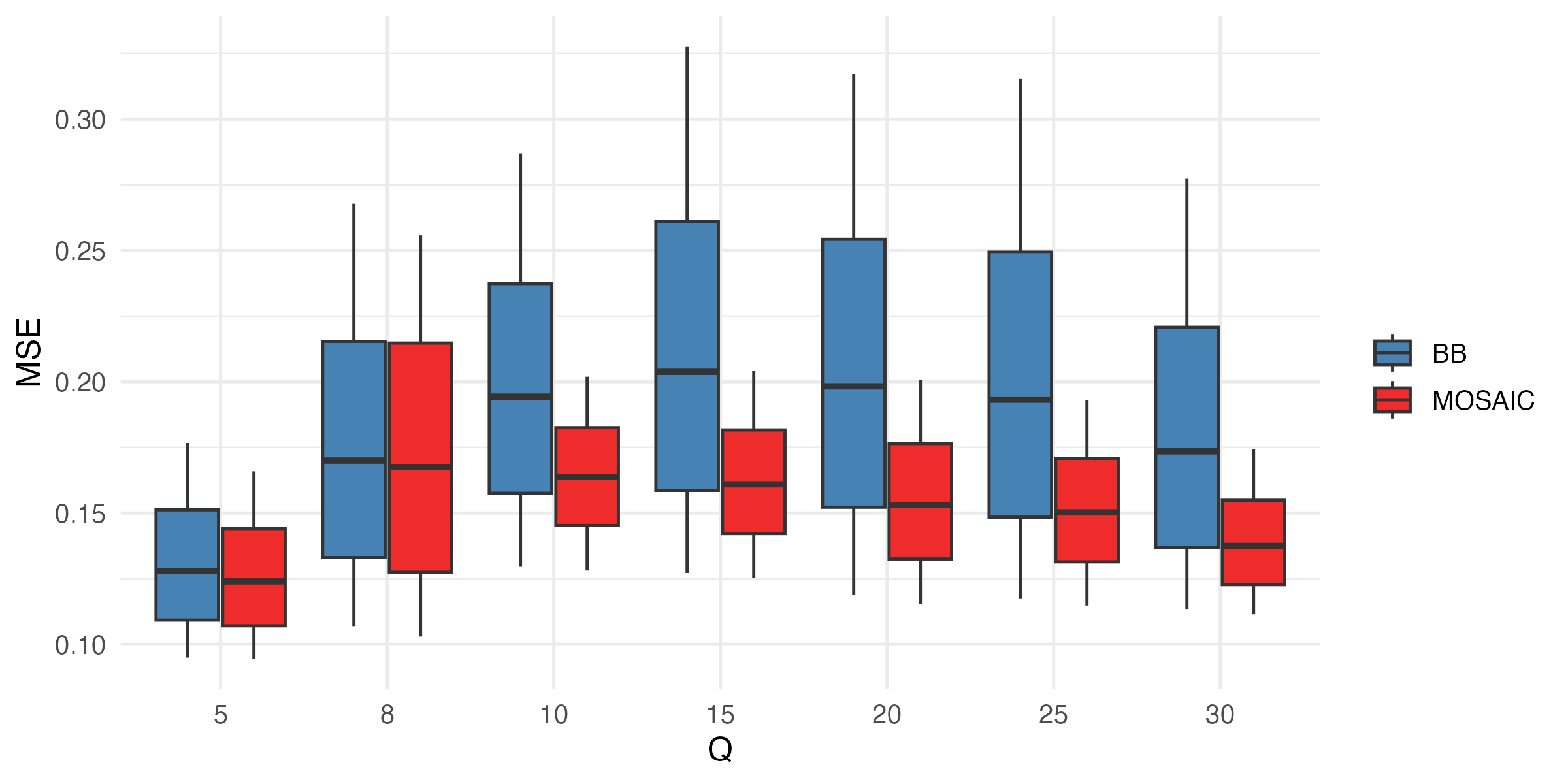}
    \caption{Average MSE in predicting $K_{q,s_q}^{(n_q)}$, averaged over $l=1,\dots,s_q$, with $s_q=10$, for the \texttt{MOSAIC} model and the exchangeable negative binomial mixture of beta-Bernoulli models (\texttt{BB}) across different values of $Q$.}
    \label{fig:QmulMSE}
\end{figure}

To assess whether Bayesian posterior distributions are well calibrated, we evaluate the frequentist coverage of posterior credible intervals for the regression coefficients $\bm{\zeta}$ and precision parameters $\bm{\phi}$. These parameters are directly interpretable in the \texttt{MOSAIC} framework - $\bm{\zeta}$ governs the mean structure, while $\bm{\phi}$ controls site-level dispersion - making reliable uncertainty quantification essential. We simulate the data under the \texttt{MOSAIC} model using the same parameter settings as in the main paper with $Q=15$ sites, but consider two different values for the mean of the negative binomial distribution from which $N$ is drawn, namely $3000$ and $5000$. We generated 500 datasets and for each replicate compute $95\%$ posterior credible intervals for each component of $\bm{\zeta}$ and $\bm{\phi}$. Coverage is then estimated as the proportion of times the true parameter values are contained within these intervals.
The results, reported in Table~\ref{tab:coverage}, show coverage close to the nominal level for both parameters, with moderate variability across components.

\begin{table}[t]
\centering
\caption{Coverage probabilities under two values of $\nu_0$.}
\label{tab:coverage}
\vspace{6pt}
\begin{tabular}{lcc@{\qquad}lcc}
\hline
Parameter & $\nu_0=3000$ & $\nu_0=5000$ & Parameter & $\nu_0=3000$ & $\nu_0=5000$ \\
\hline
$\zeta_1$ & 0.93 & 0.93 & $\phi_1$  & 0.95 & 0.95 \\
$\zeta_2$ & 0.95 & 0.95 & $\phi_2$  & 0.94 & 0.94 \\
$\zeta_3$ & 0.95 & 0.96 & $\phi_3$  & 0.90 & 0.91 \\
          &      &      & $\phi_4$  & 0.96 & 0.95 \\
          &      &      & $\phi_5$  & 0.93 & 0.93 \\
          &      &      & $\phi_6$  & 0.93 & 0.92 \\
          &      &      & $\phi_7$  & 0.95 & 0.94 \\
          &      &      & $\phi_8$  & 0.93 & 0.94 \\
          &      &      & $\phi_9$  & 0.94 & 0.94 \\
          &      &      & $\phi_{10}$ & 0.94 & 0.94 \\
          &      &      & $\phi_{11}$ & 0.93 & 0.95 \\
          &      &      & $\phi_{12}$ & 0.94 & 0.95 \\
          &      &      & $\phi_{13}$ & 0.95 & 0.94 \\
          &      &      & $\phi_{14}$ & 0.94 & 0.94 \\
          &      &      & $\phi_{15}$ & 0.96 & 0.94 \\
\hline
\end{tabular}
\end{table}

\subsection{Additional plots}

\begin{figure}[h!]
    \centering
    \includegraphics[width=0.8\textwidth]{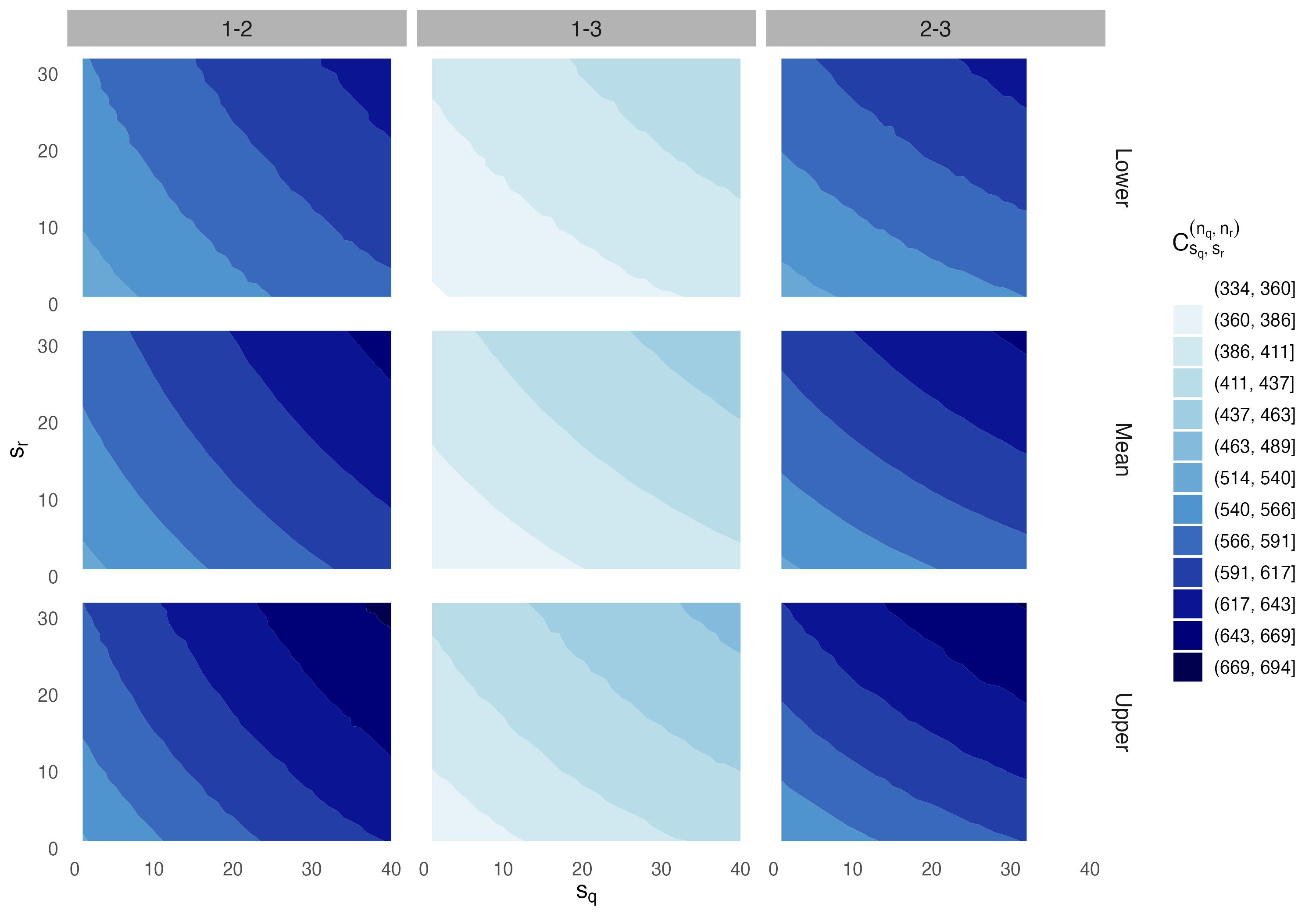}
    \caption{Posterior summaries of the rarefaction surfaces for shared species across site pairs, showing the upper bounds of the 95\% credible intervals (bottom panels), posterior means (middle panels), and lower bounds (top panels).}
    \label{fig:supp_surfUQ}
\end{figure}

\clearpage

\section{Additional details and results on fungal biodiversity application}

\subsection{Additional details on fungi data}

\begin{figure}[h!]
    \centering
    \includegraphics[width=0.95\textwidth]{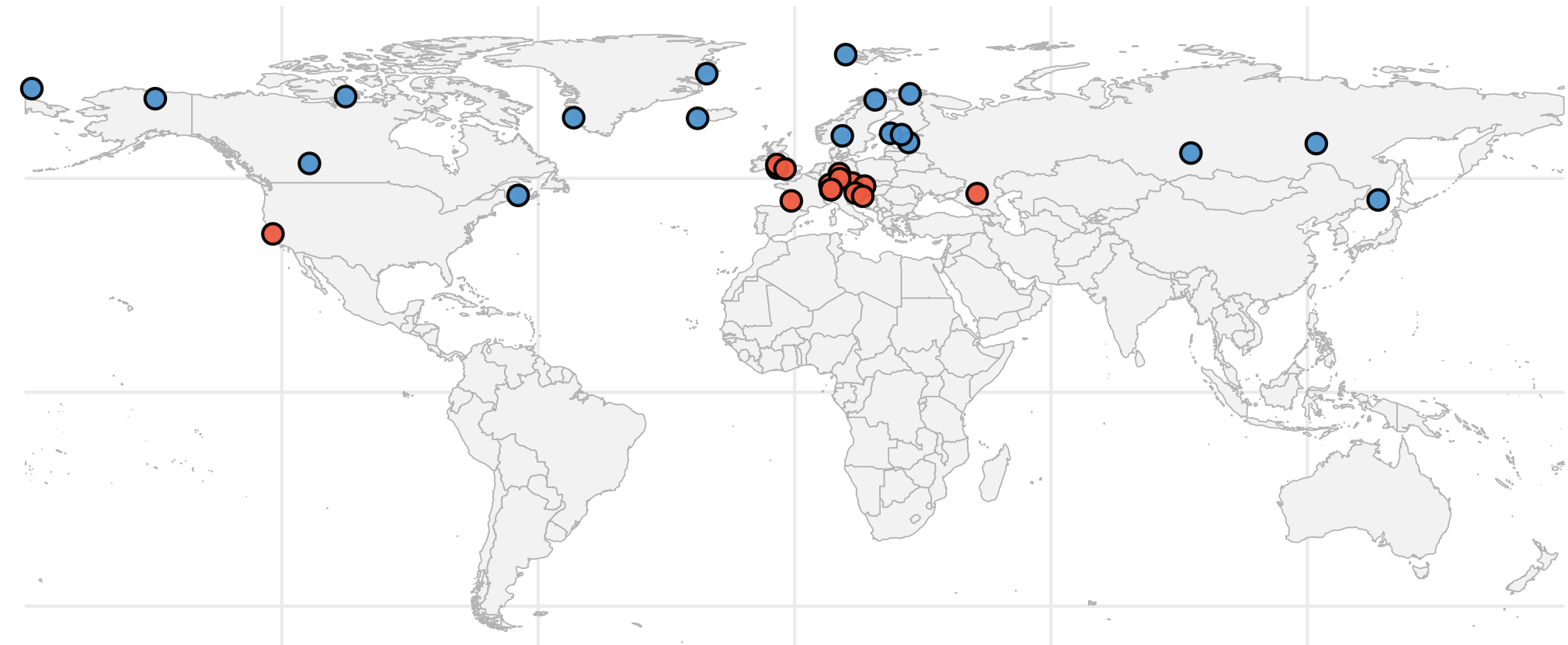}
    \caption{Map of sampling locations for the fungi data analysis colored by climatic zones.}
    \label{fig:map_site}
\end{figure}

\begin{table}[h!]
\centering
\footnotesize
\caption{Sampling locations and sample sizes for the fungi data analysis.}
\label{tab:sites}
\vspace{6pt}
\begin{tabular}{llr@{\qquad}llr}
\hline
\multicolumn{3}{c}{\textbf{Temperate}} & \multicolumn{3}{c}{\textbf{Polar-continental}} \\
Site & Name & $n_q$ & Site & Name & $n_q$ \\
\hline
ABE & Aberystwyth (Wales)     & 17  & ABI & Abisko (Sweden)        & 36 \\
BAN & Bangor (Wales)          & 17  & CAM & Sandger\o{}i (Iceland) & 12 \\
BAV & Waldh\"auser (Germany)  & 103 & CHA & Cambridge Bay (Canada) & 18 \\
CBN & Urdorf (Switzerland)    & 57  & EDM & Edmonton (Canada)      & 42 \\
EUR & Wien (Austria)          & 38  & FTS & Fetsund (Norway)       & 62 \\
HAI & Hainich (Germany)       & 93  & IGL & Memramcook (Canada)    & 16 \\
INR & Bordeaux (France)       & 45  & KEV & Kevon (Finland)        & 33 \\
KLA & Klagenfurt (Austria)    & 35  & NUU & Nuuk (Greenland)       & 27 \\
ROS & Rostov-on-Don (Russia)  & 32  & PEC & Krasnoyarsk (Russia)   & 13 \\
SAN & Santa Cruz (US)         & 34  & PRI & Wrangel Island (Russia)& 51 \\
SCH & Schwarzwald (Germany)   & 40  & PRM & Dikimdzha (Russia)     & 13 \\
STE & Steigerwald (Germany)   & 72  & SVA & Svalbard (Norway)      & 24 \\
WOR & Worcester (England)     & 71  & TAR & Tartu (Estonia)        & 78 \\
ZAG & Zagreb (Croatia)        & 69  & TOO & Toolik (Alaska)        & 55 \\
ZRC & Zurich I (Switzerland)  & 40  & TUR & Turku (Finland)        & 39 \\
ZUR & Zurich II (Switzerland) & 70  & VII & Viikki (Finland)       & 34 \\
    &                         &     & YAK & Terney (Russia)        & 24 \\
    &                         &     & ZAC & Zackenberg (Greenland) & 15 \\
\hline
\end{tabular}
\end{table}

\subsection{Additional results on fungal biodiversity application} \label{sec:supp_resFungi}

We provide rarefaction curves for the latent total number of shared species $U_{n_q,n_r}$ for pairs of sites with the same sample sizes. Figure \ref{fig:sm_rarSharedLat} shows results for three pairs under three values of $\sigma \in \{0.1, 0.4, 0.95\}$ for all sites. The results are analogous to those for site-specific rarefaction curves in the main paper: lower detectability yields notably higher expected shared species counts.

\begin{figure}[t]
    \centering
    \includegraphics[width=0.99\textwidth]{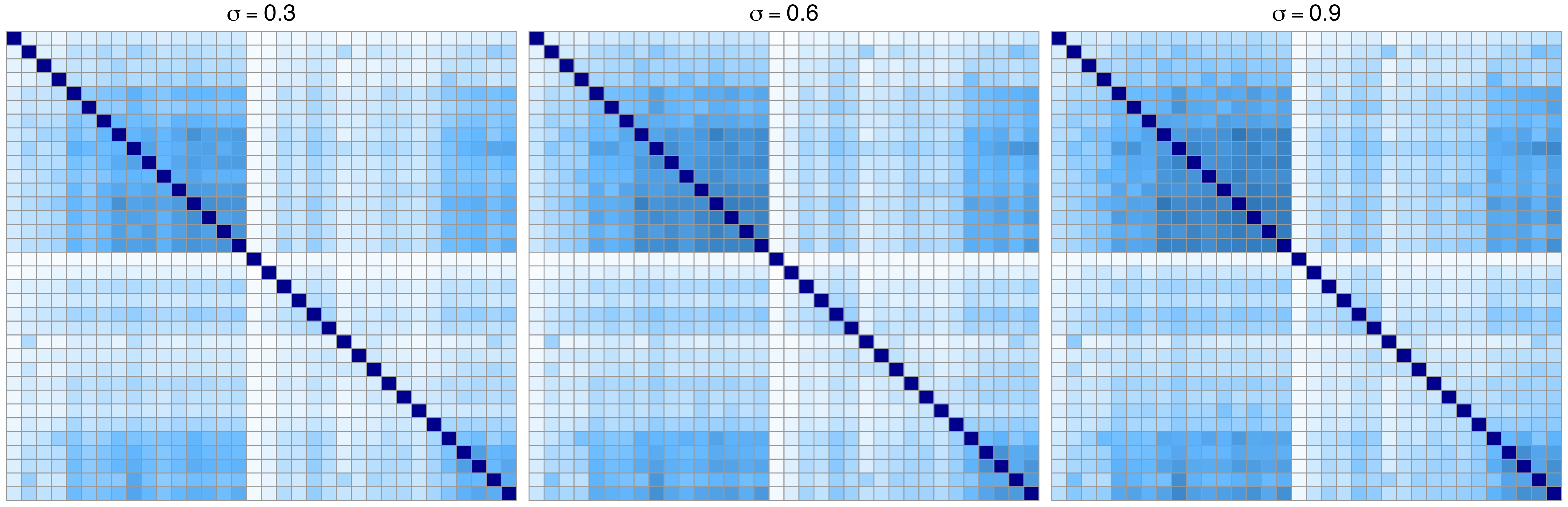}
    \caption{Posterior mean of proposed $\beta$-similarity for all sites grouped by climatic zones for different values of $\sigma_{q}=\sigma \in \{0.3, 0.6, 0.9\}$ for all $q=1,\dots,34$.}
    \label{fig:sm_betaSigma}
\end{figure}

\begin{figure}[t]
    \centering
    \includegraphics[width=0.99\textwidth]{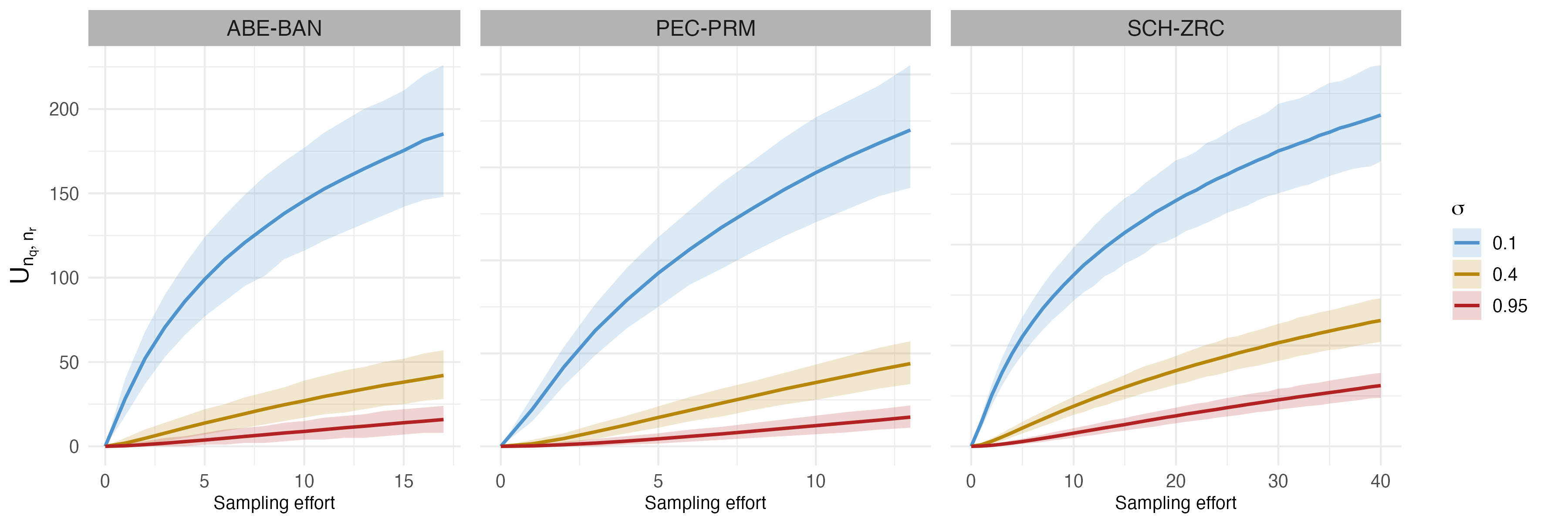}
    \caption{Posterior mean and 95\% credible intervals for simplified rarefaction curves for the latent total number of shared species $U_{n_q,n_r}$ for different values of $\sigma$.}
    \label{fig:sm_rarSharedLat}
\end{figure}

\subsection{Additional experiments for fungal biodiversity application} \label{sec:sm_fungioos}

We perform additional experiments to evaluate the out of sample predictive performance of the \texttt{MOSAIC} model for the number of shared species on the fungi data. We split the data into training and test sets, using $70\%$ of the observations in each site for training and the remaining for test, so that each training set size is proportional to $n_q$. We fit \texttt{MOSAIC} using the same hyperparameters as in the main paper on the training data.
Based on the model estimates, we then predict the number of shared species for pairs of sites in the test set, that is, $C_{s_q,s_r}^{(n_q,n_r)}$. Figure \ref{fig:fungi_oos} displays the rarefaction surfaces for shared species on the test set along with the corresponding model-based predictions for three illustrative pairs of sites: a pair of temperate locations Hainich (HAI) - Santa Cruz (SAN); a pair of polar-continental locations Viikki (VII) - Terney (YAK); and a cross-zone pair Waldhäuser (BAV) - Fetsund (FTS). The test set sizes, summarized in Table~\ref{tab:ns_oosfungi1}, vary across sites in proportion to their sample sizes, allowing evaluation of predictive performance across a range of sampling efforts.
Viikki, a district in northeastern Helsinki, and Terney, a remote coastal site in the Sikhote-Alin region of the Russian Far East, are the pair with the fewest shared species.
While both sites lie at similar latitudes, they differ markedly in vegetation and degree of human influence: Terney is characterized by mixed forests dominated by Korean pine and Mongolian oak under a climate influenced by the Sea of Japan, whereas Viikki is located within the urban research campus of the University of Helsinki.
In contrast, the two European locations - Waldhäuser and Fetsund - share the highest number of species in common.
The posterior mean closely follows the empirical surfaces for all pairs, indicating strong agreement between observed and predicted values.

\begin{table}[t]
\centering
\caption{Training ($n_q$) and test ($s_q$) sample sizes by site.}
\label{tab:ns_oosfungi1}
\vspace{6pt}
\begin{tabular}{lrrrrrr}
\hline
 & HAI & SAN & BAV & FTS & VII & YAK \\
\hline
$n_q$ & 65 & 23 & 72 & 43 & 23 & 16 \\
$s_q$ & 28 & 11 & 31 & 19 & 11 & 8 \\
\hline
\end{tabular}
\end{table}

\begin{figure}[t]
    \centering
    \includegraphics[width=0.99\textwidth]{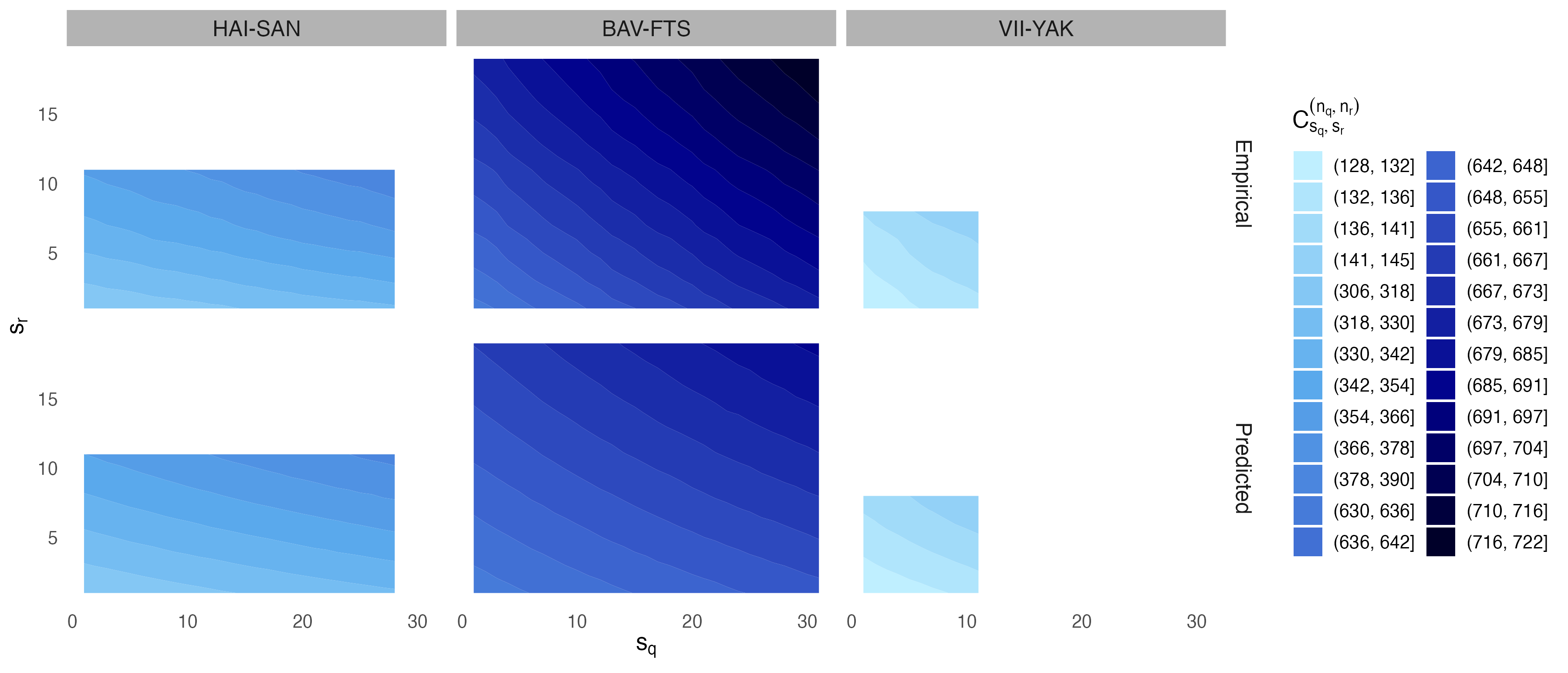}
    \caption{Empirical rarefaction surfaces for shared species in the test set, together with the corresponding posterior mean estimates of $C_{s_q,s_r}^{(n_q,n_r)}$.}
    \label{fig:fungi_oos}
\end{figure}

\end{document}